\documentclass[a4paper,12pt]{article}

\usepackage[active]{srcltx} 

\usepackage{amsmath}
\usepackage{amsthm}
\usepackage{amssymb,stmaryrd}
\usepackage[T1]{fontenc}

\usepackage{bbm}
\usepackage{amsfonts}
\usepackage{xy}           

\usepackage{graphicx}

\usepackage{url}
\usepackage{color}

\numberwithin{equation}{section}

\theoremstyle{plain}  
\newtheorem{thm}{Theorem}[section]

\newtheorem{cor}[thm]{Corollary}
\newtheorem{lemma}[thm]{Lemma}
\newtheorem{prop}[thm]{Proposition}

\newtheorem{defi}[thm]{Definition}         

\newtheorem{rem}[thm]{Remark}

 \usepackage{fancyhdr}
 \fancypagestyle{fancy2}{
\lhead[\rmfamily \thepage \fancyplain{}]{{\slshape \rightmark}}
\chead{}
\rhead[{\slshape \leftmark}]{\fancyplain{} {\rmfamily \thepage}}
\lfoot{}
\cfoot{}
\rfoot{}

		}
\renewcommand{\sectionmark}[1]{\markboth{\leftmark}{\thesection.\ {\scshape {#1}}} }

\fancypagestyle{new}{
\lhead[\rmfamily \thepage \fancyplain{}]{{\slshape \rightmark}}
\chead{}
\rhead[{\slshape \rightmark}]{\fancyplain{} {\rmfamily \thepage}}
\lfoot{}
\cfoot{}
\rfoot{}

}

\newcommand{\delete}[1]{}
\delete{
}

\delete{ \section{Counter and Comments}
}

\newcommand{\E}[0]{{\mathbbm{E}}}
\newcommand{\e}{\varepsilon}

    \newcommand{\bitem}{\begin{enumerate}}
    \newcommand{\eitem}{\end{enumerate}}

\newcommand{\lbc}[0]{\| \phi^- \|_\infty}

\newcommand{\GibK}[3]{G (#1)}
\newcommand{\tGibK}[3]{G^t (#1)}

\newcommand{\betaD}{}

\newcommand{\Pmea}[1]{\mathcal{P}_{#1}}	
\newcommand{\Pmead}[2]{\mathcal{P}_{#1}^{#2}}  
\newcommand{\PmeaK}[1]{\mathcal{G}_{#1}}  
\newcommand{\PmeaKd}[2]{\mathcal{G}_{{#2}, #1}} 
\newcommand{\Gmea}{{\mathcal{G}_\theta}} 

\newcommand{\Ggmea}[1]{\mu}

\newcommand{\pot}[0]{\phi}	
\newcommand{\potf}[2]{\pot (#1,#2)} 

\newcommand{\losy}[0]{{\mathfrak{l}}}
\newcommand{\lomagsy}[0]{\mathfrak{m}}
\newcommand{\lomagIgA}[4]{ #3 (#2)}

\newcommand{\lomagIg}[3]{ \lomagIgA {#1} {#2} {#3} {\ }}

\newcommand{\lomagf}[2]{#2 (#1)}

\newcommand{\lomagIf}[2]{\lomagIg {\hXm}{#1} {#2}}

\newcommand{\floma}{\losy}

\newcommand{\funcL}[6]{ 
    \begin{eqnarray}
			{#1} \colon& {#2} & \rightarrow\   {#3} \notag \\
   				 				&{#4} & \mapsto\   {#5}
	\label{#6}
      \end{eqnarray}}

\newcommand{\lnorm}[1]{#1 (\Xm)}
\newcommand{\lnormd}[2]{#1 (#2)}
\newcommand{\anorm}[2]{M_{#1}(#2)}

\newcommand{\oset}[2]{\overset{\text{#1}}{#2}}

\newcommand{\B}{\mathcal{B}}

\newcommand{\Css}{\Gamma}
\newcommand{\Csa}[1]{{\Css (#1)}}

\newcommand{\Csfsad}[2]{\Css_{f}(#2)}
\newcommand{\Cspa}[1]{\Css_{p}(\hXm)}

\newcommand{\Csqfg}[1]{\Css(\hXm)}
\newcommand{\Csqfd}[1]{{\Kas}(\hXd)}
\newcommand{\Csqfsg}[1]{{\Kas}(\hXm)}

\newcommand{\Kug}[1]{\Kas_{#1} ({\Xm})}
\newcommand{\Ksad}[0]{\K}

\newcommand{\Kas}[0]{\mathbbm{K}} 

\newcommand{\K}[0]{\Ka \Xm} 
\newcommand{\Ka}[1]{{\Kas(#1)}} 
\newcommand{\Ma}[0]{{\mathbbm{M}(\Xd)}} 
\newcommand{\fcone}[0]{{C_0^+(\Xd)}} 
\newcommand{\pMa}[0]{{\mathbbm{M}_+(\Xd)}} 
\newcommand{\pMah}[0]{{\mathbbm{M}_+(\hXd)}} 
\newcommand{\tKa}[1]{{\Kas^{\text{\ttfamily{t}}}(#1)}} 
\newcommand{\Kad}{\Ka \Xd} 
\newcommand{\tKad}[2]{\tKa \Xd} 
\newcommand{\M}{{\RR_+^*}}		

\newcommand{\X}{{\Xd}}
\newcommand{\Xd}{{{\RR}^d}}  
	\newcommand{\hXd}{ {\hat{\RR} ^d}}  
	\newcommand{\hXm}[0]{ \hXd}
\newcommand{\Xm}{\X}

\newcommand{\RR}{\mathbb{R}}

\newcommand{\ZZ}{\mathbb{Z}^d}
\newcommand{\NN}{\mathbb{N}}

\newcommand{\refeq}[1]{(\ref{#1})}

\newcommand{\prVer}{\today}

\delete{\part{Start}}
 \title{Gibbs states over the cone of discrete measures}
 \date{\prVer}
	
\author{Dennis Hagedorn\footnote{dhagedor@math.uni-bielfeld.de}, 
Yuri Kondratiev \footnote{kondrat@math.uni-bielefeld.de} , Tanja Pasurek \footnote{tpasurek@math.uni-bielefeld.de} , Michael Röckner \footnote{roeckner@math.uni-bielefeld.de} \\ Fakult{ä}t f{ü}r Mathematik, Bielefeld Universität, D 33615 Bielefeld, Germany}

\delete{
\author[bi]{Dennis Hagedorn\corref{cor1}\fnref{e1}}
\ead{dhagedor@math.uni-bielfeld.de}
\fntext[e1]{dhagedor@math.uni-bielfeld.de}

\author[bi]{Yuri Kondratiev\fnref{e2}}
\ead{kondrat@math.uni-bielfeld.de}
\fntext[e2]{kondrat@math.uni-bielfeld.de}

\author[bi]{Tanja Pasurek\corref{cor1}\fnref{e3}}
\ead{pasurek@math.uni-bielfeld.de}
\fntext[e3]{tpasurek@math.uni-bielfeld.de}

\author[bi]{Michael Röckner\fnref{e4}}
\ead{roeckner@math.uni-bielfeld.de}
\fntext[e4]{roeckner@math.uni-bielfeld.de}

\cortext[cor1]{Corresponding author}
}

\begin{document}				

\bibliographystyle{plain} 

\delete{\part{START}}
\maketitle

\begin{abstract}
We construct Gibbs perturbations of  the Gamma process on $\Xd$, which may be used in applications to model systems of densely distributed particles. First we propose a  definition of  Gibbs measures over the cone of discrete Radon measures on $\Xd$ and then analyze conditions for their existence. Our approach works also for general Lévy processes instead of Gamma measures. To this end, we need only the assumption that the first two moments of the involved Lévy intensity measures are finite. Also uniform moment estimates for the Gibbs distributions are obtained, which are essential for the construction of related diffusions. Moreover, we prove a Mecke type characterization for the Gamma measures on the cone and an FKG inequality for them.

Keywords: Gamma process, Poisson point process, discrete Radon measures, Gibbs states, DLR equation, Mecke identity, FKG inequality, marked configuration spaces, interacting particle systems.
\newline 
	2010 Mathematical Subject Classicfication: Primary 82B21; Secondary 28C20, 60G57, 60K35, 82B05.
\end{abstract}

\pagenumbering{arabic}

\section{Introduction}
 Analysis and stochastics on configuration spaces in the continuum may be considered, in particular, as a mathematical 
 background  of several  statistical physics models.
For example, the equilibrium states of classical free gases are given by Poisson measures (Poisson point processes) on configuration spaces.  The states of interacting  gases may be defined as Gibbs measures which are \guillemotleft singular perturbations\guillemotright\ of Poisson measures 
 in the framework of the well-known Dobrushin-Lanford-Ruelle (DLR) formalism, see,  e.g.,  \cite{dob70b,dob70,koparo10,laru69}. For complex systems  with a non-trivial internal structure of their elements (like, e.g., ecological systems in the presence of biological diversity), the notion of 
 a free system shall be specified in any particular case.  As possible candidates for the role of equilibrium states here, \textit{Lévy processes} on corresponding location spaces may be used. Then a construction of equilibrium states in the presence of interactions needs a proper generalization of the DLR approach.  The latter is the main aim of our paper.  

In particular, this framework is well-suitable to model a new class of interacting particle systems in the continuum $\Xd$, $d\in \NN$, in which to each particle $x \in \Xd$ one attaches an additional positive characteristic (mark) $s_x$ being distributed according to some \textit{infinite Lévy} measure $\lambda(ds)$ on $\M :=(0,\infty)$. A drastic difference between marks and position is reflected in the special properties of the corresponding Gibbs states. A new topological issue is that these states are supported by locally finite, postive \textit{discrete} measures on the location space $\Xd$. The cone $\K$ of such measures constitutes an intermediate \guillemotleft coordinate\guillemotright\  space between the spaces $\Gamma(\Xd)$ and $\mathbbm{M} (\Xd)$ of locally finite configurations resp. Radon measures over $\Xd$. 

\begin{minipage}[b][5 cm][b]{10 cm}
\resizebox{9.6cm}{!}{\begin{picture}(0,0)%
\includegraphics{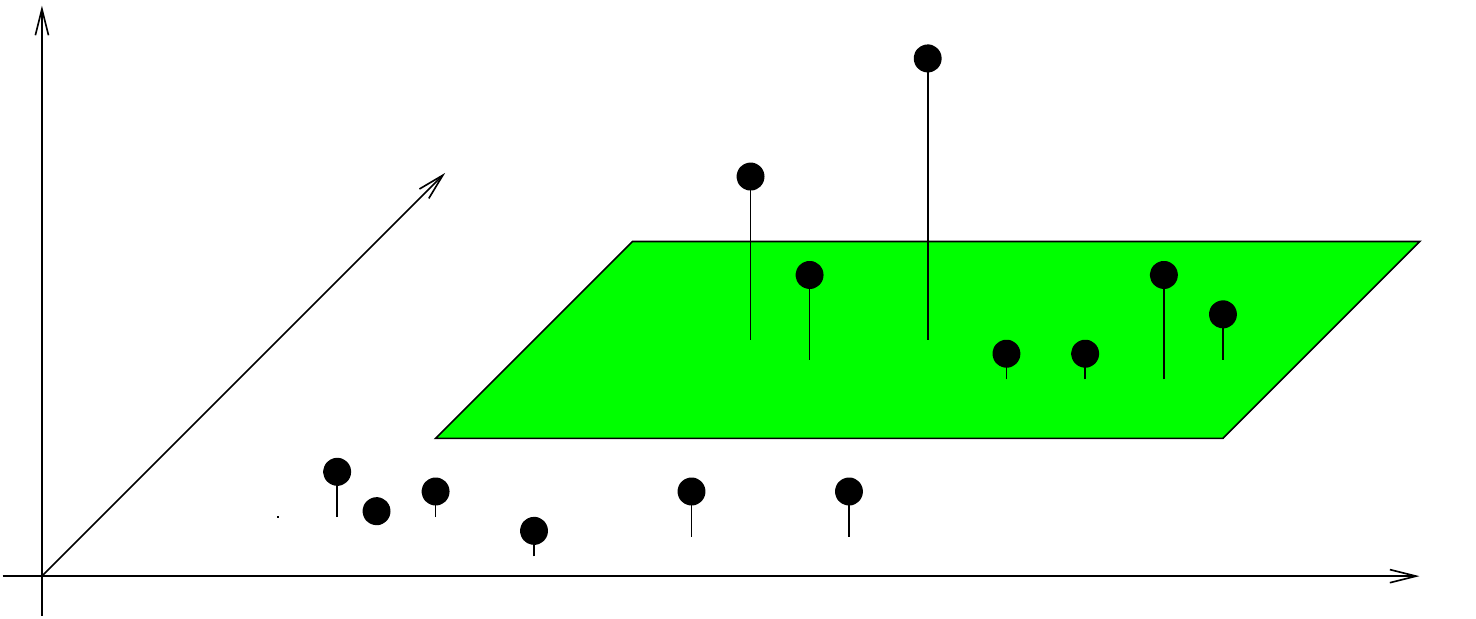}%
\end{picture}%
\setlength{\unitlength}{4144sp}%
\begingroup\makeatletter\ifx\SetFigFont\undefined%
\gdef\SetFigFont#1#2#3#4#5{%
  \reset@font\fontsize{#1}{#2pt}%
  \fontfamily{#3}\fontseries{#4}\fontshape{#5}%
  \selectfont}%
\fi\endgroup%
\begin{picture}(6732,2814)(2509,-6553)
\put(8461,-5776){\makebox(0,0)[lb]{
\textnormal{$\Delta$}
}}
\put(9226,-4606){\makebox(0,0)[lb]{
\textnormal{$\Xd$}
}}
\put(2926,-4246){\makebox(0,0)[lb]{
\textnormal{$\M$}
}}
\end{picture}%
}
\end{minipage}

Each measure $\eta \in \K$ can be written in the form 
$$
	\eta = \sum_{i} s_i \delta_{x_i} \quad \text{ with $s_i >0$.}
$$ 
Because of a possibly high concentration of the intensity measure $\lambda(ds)$ near zero, the positions of particles $x \in \Xd$ form typically a \textit{dense} countable set in $\Xd$, i.e., in each open $\Delta \subset \Xd$ there are a.s. infinitely many $x_i$'s. This is principally different from the case of marked configuration spaces with \emph{finite} measures on marks, which was mostly studied in the literature (cf., e.g., \cite{alkoro98anaGeo,alkoro98anaGibs,kosist98,kunphd} and the references therein).

Although the results of our considerations below hold for quite general Lévy processes  (cf. Section \ref{s4a}),
for simplicity, we focus here on the particular case of so-called Gamma measures corresponding to the choice of 
$$
	\lambda(ds) := \theta e^{-s}/s ds \quad \text{ with $\theta >0$.}
$$
This case is especially interesting in applications and has
several  additional analytic  properties, e.g., related with the quasi-invariance of Gamma measures.

Gamma processes on general location spaces $X$ and associated Gamma measures on the corresponding cone $\mathbbm{K}(X)$ appear in quite different areas of modern analysis and probability. One of the most impressive highlights here is an essential role of Gamma measures in the respresentation theory of big groups, as it was discovered first by Vershik, Gelfand and Graev in \cite{vegegr75}. Furthermore, Gamma measures are closely related to additive and multiplicative Lebesgue measures in infinite dimensions (cf. \cite{tsveyo01,ver07c}). On the other hand, they deliver an example of measures admitting a closed form of the analytic generating functional for orthogonal polynomials. The latter yields the corresponding chaos decomposition and leads to a well developed version of the white noice analysis (cf., e.g., \cite{kssu98}). Along with Gaussian measures on linear spaces and Poisson measures on configuration spaces, Gamma measures on the cone constitute the third prominent example of measures in infinite dimensions due to a version of Meixner's classification (cf. \cite{mei34}).
And, finally, Gamma measures (and more general compound Poisson measures) form an important class of random fields in statistics and applications to the theory of interacting particle systems. So, in \cite{sta03}, the Gamma measures appear as examples of ``invariant probability measures for a class of continuous state branching processes with immigration''\footnote{This is cited from \cite[Abstract]{sta03}.}.\\

The paper is organized as follows: In Section \ref{s2}  we introduce a family of Gamma measures $\Gmea$, $\theta>0$, on the cone $\K$ and discuss their basic properties. In Section \ref{s3}, we fix a (not necessarily nonnegative) stable pair potential $\pot(x,y)$ and define the corresponding Gibbs reconstructions $\mu \in \GibK {\pot} \theta \Xm$ of the \guillemotleft free\guillemotright\ measure $\Gmea$ as a solution of the DLR equation \refeq{B18}. We further reduce our considerations to a proper subset $\tGibK \pot \theta \Xm$ of so called \textit{tempered} Gibbs measures with controlled growth (cf. Eq. \refeq{B9c3a1}). In Section \ref{s4}  we prove the main Theorems \ref{Bthm9c4E} and \ref{Bthm4c8} on existence and à-priori moment bounds for $\mu \in \tGibK \pot \theta \Xm$. Note that, just as in most of the continuous particle systems, the existence problem for Gibbs states on the cone $\K$ is far from being evident (even in the simpler case of $\phi \geq 0$). Since the interacting potentials appearing here void, in general, the usual assumptions of integrability and translation invariance, Ruelle's technique of superstability estimates (cf. \cite{alkoro98anaGibs,mas00,rue69,rue70}) does not apply directly. Another basic method, which relies on a fundamental Dobrushin's existence criterium (cf. \cite{dob70b,dob70}), is neither applicable because of missing regularity properties of the interaction (cf. Remark \ref{rem41}). Therefore, we develop an analytic approach to the existence problem of Gibbs measures in this situation. It employs Lyapunov functionals, weak dependence of the Gibbs specification on boundary conditions and a proper topology of local setwise  convergence. On configuration spaces, this approach was first applied in \cite{koparo10} to construct Gibbs perturbations of Poisson point fields with spatially irregular intensity measures.  
	In Section \ref{s5}, we comment on the results obtained so far and outline some core extensions of the initial model. Finally, in Section \ref{s6}, we have a closer look at the intrinsic relation between Gamma and compound Poisson measures and derive a Mecke type identity and an FKG inequality for $\Gmea$.

\section{Gamma measures}\label{s2}

As a location space, let us fix the $d$-dimensional Euclidean space $(\Xd, |\cdot|)$. It is endowed with the Lebesgue measure $m(dx)$ on the Borel $\sigma$-algebra $\B(\Xd)$. By $\B_c(\Xd)$ we denote the ring of all bounded (i.e., those with compact closure) sets from $\B(\Xd)$. The continuous and compactly supported functions $\varphi: \Xd \rightarrow \RR$ form a locally convex vector space $C_0(\Xd)$, which is given a natural topology of uniform convergence on sets from $\B_c(\Xd)$. By the Riesz representation theorem, the dual space of $C_0(\Xd)$ can be identified with the space $\Ma$ of all \textit{signed Radon} (i.e., locally finite) measures on $(\Xd, \B(\Xd))$. By definition, each $\nu \in \Ma$ is finite on all $\Delta \in \B_c(\Xd)$. The space $\Ma$ will be equipped with the \textit{vague topology}, which is the coarsest topology making all mappings 
\begin{align}
	\Ma \ni \nu \mapsto \langle \varphi, \nu \rangle := \int_{\Xd} \varphi(x) \nu(dx), \quad \varphi \in C_0(\Xd),
\label{eq1}
\end{align}
continuous. It is well known (see e.g. \cite[15.7.7]{kal83}) that $\Ma$ is \textit{Polish}, i.e., there exists some separable and complete metric on $\Ma$ generating the vague topology. By $\B(\Ma)$ we denote the corresponding Borel $\sigma$-algebra on $\Ma$; the one-point sets (e.g., $\{ \nu = 0\}$) clearly belong to $\B(\Ma)$.
Let us abbreviate $\RR_+ := [0,+\infty)$ and $\M := (0,+\infty)$. By $\fcone$ resp. $\pMa$ we denote the cone of all nonnegative functions $\varphi \in C_0(\Xd)$ resp. the dual cone of all nonnegative measures $\nu \in \Ma$. \\

The Gibbs states considered below will be supported by the \emph{cone of (nonnegative) discrete Radon measures} over $\Xd $ defined as 
\begin{align}                                                          
        \Ka \Xd:= \Big\{ \eta = \sum_i s_i \delta_{x_i} \in \Ma \Big|& s_i \in \M, \thinspace x_i \in {\Xd} 
		\Big\}
     \label{2z0a}.
\end{align}
Here, $\delta_{x_i}$ are Dirac measures, the atoms $x_i$ are assumed to be distinct and their total number is at most countable. 
By convention, the cone $\Ka \Xd$ contains the null mass $\eta = 0$, which is represented by the sum over the empty set of indixes $i$.
We will refer to each $s_i$ as a \emph{mark} and to each $x_i$ as a \emph{position}. This terminology is motivated by marked configuration spaces (cf. e.g. \cite{kssu98} and Section \ref{s6} below). However, our setting does not fit in that framework because  the set of all positions of an arbitrarily chosen $\eta \in \K$, i.e., its \emph{support} 
\begin{align}
  \tau(\eta) := \{ x \in \Xd | \thinspace 0 < \eta(\{x\}) =: s_x (\eta) \},
\label{eq3}
\end{align}
is typically not a (locally finite) configuration in $\Xm$. Whenever it is clear which discrete measure $\eta \in \K$ is meant, we write for short just $s_x$ instead of $s_x(\eta)$.\\ 

The closure of $\K$ w.r.t. the vague topology is the whole space $\pMa$.
By $\B(\K)$ resp. $\B(\pMa)$ we denote the trace $\sigma$-algebra of $\B(\Ma)$ on the cone $\Ka \Xm$ resp. $\pMa$. Note (see e.g. \cite[Lemmas 2.1 and 2.3]{kal83}) that each $\nu \in \pMa$ obeys a unique decomposition $ \nu = \nu_0+\eta$ into a diffusive (i.e., non-atomic) component $\nu_0 \in \pMa$ and a discrete one $\eta \in \K$. Furthermore, the mappings $\nu \mapsto \nu_0$ and $\nu \mapsto \eta$ are measurable, which implies that $\K = \{ \nu \in \pMa | \nu_0 = 0 \} \in \B(\Ma)$ and $\B(\K) \subset \B(\pMa) \subset \B(\Ma)$. The latter also yields that $(\K, \B(\K))$ is a \emph{standard Borel space} (cf. \cite[Theorem V.2.2]{par67}).

\begin{rem} 
It is an open problem whether one can introduce a metric on $\K$ making it a Polish space and being compatible with the vague topology inherited from $\Ma$. A possible way would be to show that $\K$ is a $G_\delta$-set in $\pMa$ and then to apply the Alexandrow-Hausdorff theorem. It was shown in \cite{bara97b} that all probability measures from $\K$ constitute an $F_{\sigma, \delta}$-set in $\pMa$, which still does not solve the metrization problem. 
\end{rem}

By $\mathcal{P}(\Ma)$, resp. by $\mathcal{P}(\pMa)$ and $\mathcal{P}(\K)$, we denote the space of all probability measure on $\Ma$, resp. $\pMa$ and $\K$. (In the terminology of \cite{kal83} they are called \emph{random measures}.)\\

Our basic example of a measure on $\Ka \Xd$ is a Gamma measure. A Gamma measure $\Gmea$, $\theta > 0$ being a fixed parameter, is characterized by its Laplace transform (cf. \cite[Theorem 3.7]{bechre76}) \begin{align}
      \mathbb{E}_{\Gmea} \left[ \exp \left( - \langle \varphi, \cdot\rangle \right) \right]
     = \exp \left[ - \theta \int_{\Xd} \log (1 + \varphi(x)) m(dx) \right],
     \quad \varphi \in \fcone.
    \label{2z0b}
   \end{align}     
Note that Eq.  \refeq{2z0b} extends to any bounded, compactly supported Borel function $\varphi: \Xd \rightarrow (-1,\infty)$ (for which, of course, $\log( 1 + \varphi) \in L^1(\Xd,m)$).

\begin{rem}
\bitem \item
Fix $\Delta \in \B_c(\Xd)$. In the later proofs, we also use the cone $\Ka \Delta \in \B( \Ka \Xd)$ which consists of those discrete measures $\eta \in \Ka \Xd$ which are supported by $\Delta$. There is a canonical projection
\begin{align}
	\mathbbm{P}_\Delta: \Kad \ni \eta \mapsto \eta_\Delta:= \sum_{x \in \tau(\eta) \cap \Delta} s_x \delta_x \in \Ka \Delta
	\label{eq24b}.
	\end{align}
	Respectively, we consider the Gamma measure $\PmeaKd \theta \Delta := \PmeaK \theta \circ \mathbbm{P}_\Delta^{-1}$, which has full support on $\Ka \Delta$. It is also characterized via its Laplace transform, whose formula one obtains by replacing $\Xd$ by $\Delta$ in \refeq{2z0b}.
\item
Each Gamma measure is uniquely determined by \refeq{2z0b} because the exponents
$\K \ni \eta \mapsto \exp(\langle - \eta, \varphi \rangle )$, $\varphi \in \fcone$, constitute a measure defining class on $\B(\Ma)$.
The existence (and uniqueness) of $\mu \in \mathcal{P}(\Kad)$ with the Laplace transform \refeq{2z0b} follows by the general result \cite[Theorem 3.7]{bechre76}. Alternatively, one can apply Minlos' theorem giving the existence of the corresponding $\mu$ on the nuclear space of Schwartz distributions $\mathcal{D}'(\Xd) \supset \Ma$ and then prove that $\mu$ is indeed supported by the cone $\K$ (see e.g. \cite{gevi64d}).
\item There is an explicit construction of the Gamma measure $\Gmea$ observed first in  \cite{tsveyo01}. Namely, one deduces its existence as an image measure of a Gamma-Poisson measure $\mathcal{P}_\theta$ on the configuration space $\Gamma(\M \times \Xd)$ over the product space $\M \times \Xd$ with intensity measure $\theta \frac 1 s e^{-s} ds \otimes m(dx)$, where $dt$ is the Lebesgue measure on $\M$. For more details see Section \ref{s6}.
\item 
Consider the probability space $(\K, \B(\K),\Gmea)$. The support 
$$
	\K \ni \eta \mapsto \tau(\eta) \in \B(\Xd)
$$
can be seen as a  \emph{stationary countable dense} random set in $\Xd$. By the general theory, it has quite specific properties, see \cite{alba81,ken00}.
	Furthermore, for each nonempty $\Delta \in \B_c(\Xd)$ we have 
	\begin{align}
		\int_{\K} \left| \tau(\eta) \cap \Delta \right| \Gmea (d\eta) = \infty,
		\notag
	\end{align}
	which manifests an additional distinction from Lebesgue-Poisson measures on the configuration space $\Gamma(\Xd)$ (cf. \cite{alkoro98anaGeo} for their detailed properties).
		
\eitem
\end{rem}
From the explicit form of the Laplace transform (cf. Eq. \refeq{2z0b}) one obtains two important properties of $\Gmea$: 
\begin{itemize}
	\item All \emph{local polynomial moments} exist, i.e., for $n \in \NN$ and for each bounded Borel function $\varphi :\Xd \rightarrow \RR$  being supported by $\Delta \in \B_c(\Xd)$, we have 
	\begin{align}
        \E_\Gmea [|\langle \varphi, \cdot \rangle|^n] \leq n! {\| \varphi \|_\infty} m (\Delta)^n \theta^n < \infty
    \label{2z0c}.
    \end{align}
    For such $\varphi$, the right-hand side of Eq. \refeq{eq1} defines a measurable linear functional $\Ma \ni \nu \mapsto \langle \varphi, \nu \rangle \in \RR$. 
   \item The random measure $\Gmea$ has \emph{independent increments} (or the \emph{locality property}) in the sense that $\eta(\Delta_1), \dots, \eta(\Delta_N)$ are independent for any $N \in \NN$ and disjoint $\Delta_1, \dots, \Delta_N \in \B_c(\Xd)$. In other words, 
   \begin{align}
   	\int_{\K} \prod_{i=1}^N \varphi_i(\eta(\Delta_i)) \Gmea (d \eta)
   	= \prod_{i=1}^N \int_{\K} \varphi_i(\eta(\Delta_i)) \Gmea (d \eta)
   \label{eq27}
   \end{align}
   for any collection of $\varphi_i \in L^\infty (\RR)$, $1 \leq i \leq N$.
\end{itemize}
The latter property will be crucial for constructing Gibbs perturbations of the Gamma measure, which will be done in the next section.

	\section{Gibbsian formalism on $\K$}
\label{Bs2}\label{s3}
Fixing a proper pair potential,  we introduce the notion of related Gibbs measures via a local Gibbs specification. We proceed in the spirit of the Dobrushin-Lanford-Ruelle (\emph{\textbf{DLR}}) approach to Gibbs states in statistical physics (see ,e.g., the monograph \cite{geo88}). 

\begin{description} \item[\textit{\textbf{Assumption $\mathbf{(\pot)}$}}] Let us be given a symmetric pair potential  
\begin{align}
	\pot: \Xm \times \Xm \rightarrow \RR 
	\label{Bex7}
\end{align}
being a \emph{bounded} and $\B(\Xm \times \Xm)$-measurable function such that the following conditions hold:
		 \begin{description}		 	
		 	\item[$\mathbf{(FR)}$] 
		\textit{\textbf{Finite} \textbf{range}\textbf{:}}
	   \textit{There exists $R\in (0,\infty )$ such that} 
   \begin{align}    
      \phi (x ,y) = 0, \quad \text{ if $|x-y| > R$}
      \notag. 
    \end{align}
   \item[$\mathbf{(LB)}$] \textit{\textbf{Lower bound constant}}: 
	  	\begin{gather}
				\lbc  := - \inf_{ x, y \in \Xd} 
					\{ \phi(x, y) \wedge 0 \}
					  < \infty
				\notag. 
			\end{gather}
		\item[$\mathbf{(RC)}$] \textit{\textbf{Repulsion condition}:}		
	  		\textit{There exists $\delta > 0$ such that} 
				\begin{align}				  			 
				 A_\delta := \inf_{\substack{x, y \in \Xd \\ | x - y | \leq \delta}} 
					\phi(x,y)
					>  
					2 {m^{\phi}_\delta \lbc }
						\label{Bex20a},
				\end{align}
				\textit{with interaction parameter (cf. \refeq{eq34b} below)}
				\begin{align}
					m^{\phi}_\delta := \nu_d d^{d/2} \left[ R/\delta +1 \right]^d
				\label{Bex21c},
				\end{align}    
				\textit{where $\nu_d := \frac {\pi^{d/2}} {\Gamma(d/2 +1)}$ is the volume of the unit ball in $\Xd$.}
    \end{description}
    \end{description}
		Merely speaking, the relation \refeq{Bex20a} means that the repulsion part $\phi^+ := \phi \vee 0$ of $\phi$ dominates its  attraction part $\phi^- := - \phi \vee 0$.
				Note that neither \emph{translation invariance} nor\emph{\ continuity}
		of $\pot$ need to be assumed.

		\subsection{Partition of the space $\Xm$}\label{Bss2b}     
Let $\delta >0$ be such that the repulsion condition $\mathbf{(RC)}$ holds and define the parameter $g:=\delta /\sqrt{d}$. Consider the cubes indexed by $k \in \ZZ$
    \begin{align}
     {Q_k} :=&   	 
    			\big[ - 1/2 g, 1/2 g \big)^d +  g k \subset \Xd 
    \notag
    \end{align}%
constituting a partition of $\Xd$. Each cube ${Q_k}$ is centered at the point $g k$ and has edge length $g>0$, Lebesgue volume $m(Q_k)= g^d$ and diameter
		\begin{equation}
			\mathrm{diam}\left({Q_k}\right)
				:=
				\sup_{{} x, {} y \in {} Q_k}
							|x - y|_{\Xd} = \delta
				\notag.
		\end{equation}%
The latter implies that $\phi( x, y) \geq A_\delta$ 
		for all ${} x, {} y \in {{} Q_k}$.
		To explain the choice of the constant $m_\delta^\phi$ in \refeq{Bex21c}, we introduce some more concepts and notation. 
For each $k \in \ZZ$, the family of \guillemotleft\emph{neighbor}\guillemotright\  {cubes} of ${} Q_k$ (i.e., those ${} Q_j$, $j \neq k$, having a point $y \in {} Q_j$ that interacts with a point $x \in {} Q_k$) is indexed by 
	 \begin{align}
        \mathrm{\partial }^{\phi}_\delta k
        :=
            \left\{ j \in \ZZ \backslash \{k\} \ | \ 
              	\exists x \in {{} Q_k}, \thinspace \exists y \in {{} Q_j}: 
            	\quad \phi(x,y) \neq 0 
        \right\}  
    \label{B21bc}.
    \end{align}    
   	The number of such \guillemotleft\emph{neighbor}\guillemotright\ {cubes} for every ${} Q_k$, $k \in \ZZ$, can be roughly estimated by 
   	\begin{gather}
    	  	\sup_{k \in \ZZ} 
    	  		|\mathrm{\partial }^{\phi}_\delta k| \leq m^{\phi}_\delta,
    	  \label{eq34b}
    	\end{gather}	
where $m^{\phi}_\delta$ was defined in \refeq{Bex21c}.
  		
    To each index set $\mathcal{K}\Subset \ZZ$ (this notation means that $\mathcal{K}$ is a non-void \emph{finite} subset of $\ZZ$) there corresponds
    \begin{equation}
    \Delta _{\mathcal{K}}:=\bigsqcup_{k\in \mathcal{K}}{{} Q_k}
    	\in \B(\Xm)
   \label{eq34c};
    \end{equation}   
    the family of all such domains is denoted by \emph{$\mathcal{Q}_{c}(\Xm)$}. Respectively, for $\Delta \in \B(\Xm)$ we define 
    \begin{align}
    	\begin{array}{l}
    		{\mathcal{K}_\Delta} := \{ j \in \ZZ \ | \ {} Q_j \cap \Delta \neq \varnothing\};
    	\end{array}
    \label{eq34d}
    \end{align}
	then $|\mathcal{K}_{\Delta }|$ is the number of {cubes} ${{} Q_k}$ having non-void intersection with $\Delta$.
Note that
	\begin{align}
		 |{\mathcal{K}_\Delta}| < \infty, 
		 		\qquad \forall \Delta \in \B_c(\Xm)
		\notag.
	\end{align}

	\subsection{Local Gibbs specification}
	\label{Bss2}
For each $\eta = \sum_{x \in \tau(\eta)} s_x \delta_x, \thinspace \xi= \sum_{y \in \tau(\xi)} s_y \delta_y \in \K$ and $\Delta \in \mathcal{B}_c(\Xm)$, we define the \emph{relative energy} (\emph{Hamiltonian})
\begin{align}
     H_\Delta(\eta|\xi)
    &:=
     \int_{\Delta} \int_{\Delta} \phi(x,y) \eta(dx) \eta(dy)
       +
	   2 \int_{\Delta^c} \int_{\Delta} \phi(x,y) \eta(dx) \xi(dy)
  \label{3c0a}.
\end{align}
In the particle picture, the Hamiltonian can be written as 
\begin{align}
		H_\Delta (\eta|\xi) 
		= 
     	\sum_{x,x' \in \tau(\eta) \cap \Delta} \phi(x,x') s_x s_{x'}
			 + 2 \sum_{\substack{x \in \tau(\eta) \cap \Delta \\ y \in \tau(\xi)\cap{\Delta^C}}}
					\phi(x,y) s_x s_y.
 	\notag
 	\end{align}
 	\begin{lemma}\label{lemBa1a}
	  The relative energy is finite, i.e., 
	\begin{align}
		|H_{\Delta} (\eta|\xi)| 
		< \infty, \quad \text{for all $\eta, \xi \in \Ksad  $ and $\Delta \in \B_c (\Xm)$.}	
	\notag
	\end{align}
	\end{lemma}
	
	\begin{proof}
	Note that 
	\begin{align}
		H_\Delta (\eta|\xi) 
		\leq 
		\lomagf \Delta \eta \lomagf \Delta \eta \|\phi\|_\infty + 2 \lomagf \Delta \eta \lomagf {{\mathcal{U}_\Delta}} \xi \|\phi\|_\infty,
		\notag
	\end{align} 
		where 
	\begin{align}
	      \mathcal{U}_\Delta :=& 
	      	\bigsqcup_{k \in \ZZ} \left\{  {{} Q_k} \ \big| \  
	      			\partial^{\phi}_\delta k \cap {\mathcal{K}_\Delta} \neq \varnothing \right\} \cap \Delta^c
	      		\quad \in \mathcal{B}_c(\Xm)
		\notag
	 \end{align}  
	 and $\mathcal{K}_\Delta$ was defined in \refeq{eq34d}. Since $\eta, \xi \in \Ksad $, the assertion follows. 
	 \end{proof}
	 \delete{\subsubsection{Lyapunov functional}}
	 To show upper and lower bounds for the partition function, the following \guillemotleft stability\guillemotright\ estimate for local Hamiltonians is essential.    
	 Everywhere below, we suppose that Assumption \emph{\textbf{$\mathbf{(\pot)}$}} holds with a fixed $\delta >0$ and in general omit the index $\delta$. 
	 
    \begin{lemma}\label{Bl2a}
     Let Assumption \emph{\textbf{$\mathbf{(\pot)}$}} hold. Then for each $\eta, \xi \in \Kad$ and $\Delta \in \B_c(\Kad)$
	   \begin{align}
	   	H_\Delta(\eta|\xi) 
    	\geq
    	   \left[ A - 2  m^\phi \lbc \right]  \sum_{\substack{j \in {\mathcal{K}_\Delta}}}
						\lnormd {\eta_\Delta} {Q_j }^2
		  			-  m^\phi \lbc  
	    			\sum_{l \in  \mathcal{K}_{\mathcal{U}_\Delta}} 
									\lnormd {\xi_{\Delta^c}} {Q_l}^2
		 \label{BeqB9aa}.
	   \end{align}
	   More precisely, we have for each $k \in \mathbbm{Z}^d$
	   	   \begin{align}
	   	H_{Q_k}(\eta|\xi) 
    	\geq
    	  \left[A -  m^\phi \lbc \right] \lnormd {\eta} {Q_k }^2
		  			- \lbc 
	    			\sum_{j \in  \partial^\phi k} 
									\lnormd {\xi_{{Q_k}^c}} {Q_j}^2
		 \label{37b}
	   \end{align}
		 	and, choosing $\xi=0$,
	   \begin{align}
	   		H_{{{} Q_k}} (\eta_k) := H_{Q_k}(\eta_k|0) \geq   \left[ A - 2  m^\phi \lbc \right]  \lnormd {\eta} {Q_k}^2
	   	\label{B28h}.
	   \end{align}
	  \end{lemma}
	  
	  \begin{proof}
    By definition and obvious calculations
      \begin{align}
    	    H_\Delta(\eta |\xi)
    	=  H_\Delta(\eta_\Delta| \xi_{\Delta^c}) 
    	=&
    		   \sum_{\substack{j \in {\mathcal{K}_\Delta}
    												 \\ l \in {\mathcal{K}_\Delta}
    										}}
    		\int_{Q_j} \int_{Q_l} 
    			\potf x y {\eta_\Delta} (dx) {\xi_{\Delta^c}}(dy)
    	 \notag \\ &
    	 + 2   \sum_{\substack{j \in {\mathcal{K}_\Delta}\\ l \in \mathcal{K}_{\Delta^c}}}
  	\int_{Q_l} \int_{Q_j}  	\potf x y {\eta_\Delta}(dx) {\xi_{\Delta^c}}(dy)
    \label{eq18}.
    \end{align} 
    By $\mathbf{(LB)}$ and $\mathbf{ (RC)}$ the right-hand side of Eq. \refeq{eq18} is not less than 
    \begin{align}
     	 &    {A} \sum_{\substack{j \in \mathcal{K}_\Delta}}
    				\lnormd {\eta_\Delta} {Q_j }^2
    			-    \lbc
    					\sum_{\substack{j \in {\mathcal{K}_\Delta}}} 
    					\sum_{\substack{ l \in {\mathcal{K}_\Delta} \cap {\partial^{\phi}}  j}}
    					\lnormd {\eta_\Delta} {Q_j } \lnormd {\eta_\Delta} {Q_l }
    		\notag \\ &	\qquad	
    		- 2   \lbc \sum_{j \in {\mathcal{K}_\Delta}} 
    		\sum_{l\in \mathcal{K}_{\Delta^c} \cap  {\partial^{\phi}}  j} 
    		\lnormd {\eta_\Delta} {Q_j}
    				\lnormd {\xi_{\Delta^c}} {Q_l}
    	\notag. \\
    	\intertext{Because of \refeq{B21bc} and the elementary inequality $a b \leq 1/2 (a^2 + b^2)$ for $a,b \geq 0$, the above term is bounded below by}
   	&			    {A} \sum_{\substack{j \in {\mathcal{K}_\Delta}}}
    					\lnormd {\eta_\Delta} {Q_j}^2
    		-  \lbc \sum_{\substack{j \in {\mathcal{K}_\Delta}}}
	    					{m^{\phi}} \lnormd {\eta_\Delta} {Q_j}^2
	    	\notag \\ & 
	    			- {  \lbc}
	    					\bigg( 
	    					 \sum_{\substack{j \in {\mathcal{K}_\Delta}}}
									{m^{\phi}} \lnormd {\eta_\Delta} {Q_j}^2
								+  \sum_{j \in {\mathcal{K}_\Delta}} 
									\sum_{l \in  \mathcal{K}_{\Delta^c} \cap {\partial^{\phi}}  j}
									\lnormd {\xi_{\Delta^c}} {Q_l }^2
							\bigg)
	\notag \\ =&					 
						  \left[ A -  2  m^\phi \lbc  \right]
						  \sum_{\substack{j \in {\mathcal{K}_\Delta}}}
						\lnormd {\eta_\Delta} {Q_j}^2
		  			-  {  \lbc}
							 \sum_{j \in {\mathcal{K}_\Delta}} 
									\sum_{l \in  \mathcal{K}_{\Delta^c} \cap {\partial^{\phi}}  j}
									\lnormd {\xi_{\Delta^c}} {Q_l }^2
			\label{BeqB9a1}.
    	\end{align}     	
     	By \textbf{(FR)} and \refeq{eq34b}, the last summand in \refeq{BeqB9a1} dominates
   	$$
   		-   m^\phi \lbc 
   		\sum_{l \in  \mathcal{K}_{\mathcal{U}_\Delta}} 
									\lnormd {\xi_{\Delta^c}} {Q_l }^2.
    $$    
    Note that for $\Delta= Q_k$ we even have
    \begin{align}
     H_{Q_k} (\eta_k|\xi) \geq 
     	A |\eta_k|^2 - 2 \lbc |\eta_k| \sum_{l \in \partial^{\phi} k} |\xi_{Q_k^c}(Q_l)| 
     	\notag \\ \geq
     		\left[ A -  m^\phi \lbc \right] |\eta_k|^2 - \lbc \sum_{l \in \partial^\phi k} \xi_{Q_k^c} (Q_l)^2
     	\notag,
    \end{align}
    which proves \refeq{37b} and \refeq{B28h}.
		\end{proof}
		
\delete{		\subsubsection{Partition function}}
		For each $\Delta \in \B_c (\Xm)$ and $\xi \in \Ksad$, we define the \emph{partition function}
 		\begin{equation}
			Z_{{\Delta}}(\xi) :=
    		\int_{ \Ka {\Delta}} \exp \left\{
    						- \betaD H_{\Delta} (\eta_{\Delta}|\xi)
    					\right\}
    		 \PmeaKd \theta {\Delta}(d\eta_{\Delta}).
		\notag
	\end{equation}		
		\begin{lemma}\label{Bl2b}
		Let Assumption \emph{\textbf{$\mathbf{(\pot)}$}} hold. For any ${\Delta} \in \B_c (\Xm)$ and $\xi \in \Kad$ 
		\begin{align}
			0 <& Z_{\Delta}(\xi) 
   		< \infty.
			\notag
			\end{align}			
If $\phi \geq 0$, then obviously $Z_\Delta(\xi) \leq 1$.				
		\end{lemma}
	
		\begin{proof}
			We define for each ${\Delta} \in \mathcal{B}_c (\Xm)$ and $\eta, \xi \in \Ksad$
			\begin{align}
					H^+_{\Delta} (\eta|\xi)
				:=&
				\int_\Delta \int_\Delta 
				\phi^+ ({} x, {} y)  
				\eta(dx) \eta (dy)
				+ 2
				\int_{\Delta^c} \int_\Delta 
				\phi^+ ({} x, {} y) 
				\eta(dx) \xi(dy)
			\notag,
			\end{align}
			where $\phi^+ := \phi \vee 0$. 
			By Jensen's inequality  
			\begin{align}
					Z_{\Delta} (\xi)
				\geq&
					\int_{ \Ka {\Delta}} \exp \left\{ -\betaD H^+_{\Delta} (\eta|\xi) \right\}
					 \PmeaKd \theta {\Delta}(d\eta)
				\geq  
					\exp \left\{ - \betaD \int_{ \Ka {\Delta}} H^+_{\Delta} (\eta| \xi) \PmeaKd \theta {\Delta} (d\eta) \right\}
			\notag \\  
				\geq & 
						\exp \left\{ - \betaD \|\phi\|_\infty \int_{\Ka \Delta}
								 \Big[ \lomagf {\Delta}  \eta ^2  
						  + 2 \lomagf {\Delta} \eta
						  \lomagf {\mathcal{U}_{\Delta}} {\xi_{\Delta^c}} \Big] 
						 \PmeaKd \theta {\Delta}(d\eta) \right\}
				\notag.
			\end{align}
		Using \refeq{2z0c}\label{Brem10a14}, we get that
			\begin{align}
			Z_{\Delta}(\xi)
				\geq
						\exp\bigg\{ &- \betaD \|\phi\|_\infty \Big[
								\left( C_\Delta\right)^2
						  + C_\Delta \cdot 
						\lomagIf {\mathcal{U}_{\Delta}} {\xi_{\Delta^c}}
						 \Big] \bigg\}
						 >0
				\label{Beq10ab}
			\end{align}
			with 
					\begin{align}
						C_\Delta
					:=&
						2 \theta \left| {\mathcal{K}_{\Delta} }\right| g^d < \infty
				\label{eq311b}.
				\end{align}	
			Moreover, by \refeq{BeqB9aa} we deduce an upper bound 
    \begin{align}
    	Z_{\Delta}(\xi) 
    	\leq&
    		\int_{\K} \exp 
	    			\bigg\{
	    					- \left[ A - 2  m^\phi \lbc \right] 
	    					\sum_{\substack{j \in {\mathcal{K}_{\Delta}}}}
						\lnormd {\eta_\Delta} {Q_j}^2 						
			 			\bigg\}
					 \PmeaKd \theta {\Delta}(d\eta_{\Delta})
			\notag \\ &\quad \times
    		\exp \bigg\{	
						 {\betaD {m^{\phi}} \lbc 
							}  \sum_{l \in  {\mathcal{K}_{\mathcal{U}_{\Delta}}} }
									\lnormd {\xi_{\Delta^c}} {Q_l }^2
							\bigg\}
				<\infty
			\label{311c},
			\end{align}				
		which completes the proof.
		\end{proof}

\delete{\subsubsection{Local specification}  }
For each $\Delta \in \mathcal{B}_c(\Xm)$, the \emph{local Gibbs measures} with boundary conditions $\xi \in \Kad$ are given by 
\begin{align}
		\mu_\Delta(d\eta|\xi) &:= 
				\frac 1 {Z_\Delta (\xi)}
				e^{-\betaD H_\Delta(\eta|\xi)} \PmeaKd \theta \Delta (d\eta).
			\notag
	\end{align}
	Lemma \ref{Bl2b} guarantees that each $\mu_\Delta(d\eta|\xi)$ is well-defined as a probability measure on $\Ka \Delta$.

    \begin{defi}\label{def3c1}
    The \emph{local specification} $\mathit{\Pi }=\{\pi _{\Delta }\}_{\Delta
    \in \mathcal{B}_{c}(\Xm)}$ on $\K$ is a family of stochastic kernels
    \begin{equation}
    \mathcal{B}(\K)\times \K \ni (B,\xi )\mapsto \pi
    _{\Delta }(B|\xi )\in \lbrack 0,1]  \label{3c1a}
    \end{equation}%
    given by $\pi _{\Delta }(B|\xi ):=\mu_{\Delta }(B_{\Delta ,\xi }|\xi )$, where 
    \begin{align}
        B_{\Delta ,\xi }&:=\left\{ \eta _{\Delta }\in K(\Delta) \left\vert \text{\thinspace }\eta _{\Delta }\cup \xi _{\Delta ^{c}}\in
        B\right. \right\} \in \mathcal{B}(K(\Delta)). \notag
    \end{align}
  \end{defi}

\begin{rem}
	The family \refeq{3c1a} obeys the \emph{consistency} (or \emph{Markovian}) \emph{property}, which means that for all $\Delta, \tilde {\Delta} \in \mathcal{B}_c(\Xm)$ with ${\tilde {\Delta}} \subseteq {\Delta}$
	\begin{equation}
		\int_{\K}
				\pi _{{\tilde {\Delta}} }(B|\eta )
		\pi _{{\Delta} }(\mathrm{d}\eta |\xi )
		=\pi _{{\Delta} }(B|\xi ),
		\label{B16a}
		\end{equation}
	for all	$B\in \mathcal{B}(\K)$ and $\xi \in \K$. 
	By the additive structure of the relative energy (cf. Eq. \refeq{3c0a}) and the independency  property of the Gamma measure (cf. Eq. \refeq{eq27}), this property immediately follows by the construction of the family $\Pi$ (cf. \cite[Proposition 6.3]{pre76} or   \cite[Proposition 2.6]{pre05}).
	\end{rem}
	
\delete{\subsubsection{Gibbs measures}}

\begin{defi}\label{Bdef1} 
		A probability measure $\mu$ on 
    $\K$ is called a \textbf{Gibbs measure} (or \textbf{state}) with pair potential $\pot$ if it satisfies the \emph{Dobrushin-Lanford-Ruelle} (\textbf{DLR}) equilibrium equation 
    \begin{equation}
				\int_{\K} \pi _{{\Delta}}(B|\eta
	    )\mu (\mathrm{d}\eta )=\mu (B)
		\label{B18}
    \end{equation}%
    for all ${\Delta} \in \B_c (\Xm)$ 
    and $B\in \B(\K)$. The associated set of all Gibbs states will be denoted by ${ \GibK {\pot} \theta \Xm}$.\smallskip
 \end{defi}
 
 We will mainly be interested in the subset $\tGibK \pot \theta \Xd$ of \emph{tempered} Gibbs measures which are supported by  $\tKad \floma \Xm$, i.e., 
 \begin{align}
 	\tGibK \pot \theta \Xd := \tGibK \pot \theta \Xd \cap \mathcal{P} (\tKad \floma \Xm), 
 	\label{eq314b}
 	\end{align}
 	where
 	\begin{align}
 		\mathcal{P}(\tKad \floma \Xm) := \{ \mu \in \mathcal{P}(\Kad) | \mu (\tKad \floma \Xm)=1 \}.
 	\end{align}
 	Here, we define the set of \emph{tempered} discrete Radon measures by
		\begin{align}
			\tKad \floma \Xm :=&  
				\bigcap_{\alpha>0} \Kug {{\alpha}},
				\label{B9c3a1} 
		\end{align}
		where
		\begin{align}
			\Kug {{\alpha}}  :=&
					\left\{ \eta \in \K \Big| 
							\anorm { \alpha} \eta
							< \infty \right\}
			\in \B(\Kad)
			\label{B9c3a1BB},
		\end{align}	
		and
		\begin{align}
		 \anorm { \alpha} \eta
					:= \left( \sum_{k \in \ZZ} \lnormd {\eta} {Q_k}^{2} e^{-\alpha | k |}
						\right)^{1/2}
		\label{36e}.
		\end{align}
		Note that  $M_\alpha$ extends to a seminorm on the vector space of all \emph{tempered} Radon measures from $\Ma$.
		The above property means that the tail of the discrete measure $\eta$ does not increase too much.

 Clearly, by the consistency poperty it is sufficient to check the DLR equation just for some order generating (i.e., ordered by inclusion and exhausting the whole $\Xd$) sequence $\{ \Delta_N \}_{N \in \NN} \subset \mathcal{Q}_c(\Xd)$.
		
\section{Existence for general potentials}\sectionmark{Existence}\label{Bs10b}\label{s4}
The existence problem for Gibbs measures of continuous particle systems is by non mean trivial, even in the (more familiar) case of configuration spaces, each of which admits a Polish metrization.

\begin{rem}\label{rem41}
Obviously, we can transform the existence problem of a Gibbs measure on $\K$ to the configuration space $\Gamma(\M \times \Xd)$ over $\M \times \Xd$ (cf. Section \ref{s6} for the precise definition) via the canonical map
$$ 
\K \ni \eta \mapsto \sum_{x \in \tau(\eta)} \delta_{(\eta(\{x\}),x)} \in \Gamma(\M \times \Xd).
$$
Anyhow, the existence problem cannot be covered by the results previously known on (marked) configuration spaces (\cite{alkoro98anaGibs,kunphd,kokusi98, mas00}). First of all, the involved potential 
$$
	\tilde \phi \big( (t,x), (s,y) \big) := t s \phi(x,y), \quad (t,x), \thinspace (s,y) \in \M \times \Xd,
$$ 
is not translation invariant on $\M \times \Xd$, so Ruelle's approach \cite{rue69,rue70} is not applicable. Even the usual uniform integrability condition (cf. \cite{alkoro98anaGibs,kunphd}) is void: In the simplest case of $\pot(x,y) := a(|x-y|)$ with a nonzero $a \in C_0(\RR)$ we already get
\begin{align}
	\underset {(s,x) \in \M \times \Xd} {\textit{ess sup}} \int_\M \int_\Xd \left| e^{- s t \potf x y } - 1 \right| \lambda_\theta (dt) \otimes m(dy) 
	= \infty
	\notag.
\end{align}
Secondly, $\tilde \phi$ has an infinite interaction range in $\M \times \Xd$ and corresponds to  infinite measures on marks, whereas only finite ones on marks seem to be covered so far (see e.g. \cite{kunphd,mas00}). The crucial difference to these works is that the projection of a marked configuration to the position space $\Xd$ is again a locally finite configuration, whereas in our case it is not anymore, due to the infinite Lévy measure $\theta \frac 1 t e^{-t} dt$ on the marks. For $\Gmea$-almost all $\eta \in \K$, the projection is even a dense (support) set $\tau(\eta) \subset \Xd$.
\end{rem}

In order to prove the existence of $\mu \in \tGibK \pot \theta \Xd$ an essential technical
step is to check that the net of Gibbs specification kernels $\{ \pi_{{\Delta}}( d\eta |\xi) | {\Delta} \in \mathcal{Q}_c (\Xm)\}$, with a fixed tempered boundary condition $\xi \in \tKad \floma \Xm$,  is locally equicontinuous. To that end, for each $k \in \ZZ$ we introduce  the map $\K \ni \eta \mapsto \lambda \lnormd {\eta} {Q_k} ^2$, which will play the role of a \emph{Lyapunov functional} in our theory, and show that it is exponentially integrable (cf. Lemma \ref{Bl2}). After that we establish a \emph{weak dependence} of the specification kernels on boundary conditions, which is stated for small volumes $\Delta := Q_k$ in Lemma \ref{Bl4} and in the \guillemotleft thermodynamic \guillemotright\ limit $\Delta \nearrow \Xd$ respectively in Proposition \ref{Bp3}. More precisely (cf. Eq. \ref{B35}), for each $\lambda \in [0, A -  m^\phi \lbc ]$ we can find $\mathcal{C}_\lambda > 0$ such that uniformly for all $k \in \ZZ$ and $\xi \in \tKad \floma \Xm$
$$
	\limsup_{\mathcal{K}\nearrow \ZZ}
					\int_{\K} 
						\exp \left\{ \lambda \lnormd {\eta} {Q_k}^{2}\right\} 
					 \pi _{\mathcal{K}}( d \eta |\xi )
				\leq 
					\mathcal{C}_\lambda.
$$
By the weak dependence, we deduce first the mentioned local equicontinuity of the Gibbs specification (cf. Proposition \ref{Bp9c4c}) and  then the existence of a tempered Gibbs measure (cf. Theorem \ref{Bthm9c4E}).
	As a preliminary step, we check that $\exp \{ \lambda \lnormd {\eta} {Q_k} ^2\}$ is integrable w.r.t. all $\pi_{\Delta}(d\eta|\xi)$, $\Delta \in \B_c(\Xd)$.
	\begin{lemma}
    \label{Bl2} 
	Fix $k\in \ZZ$, $\xi \in \Ksad  $, ${\Delta} \in \mathcal{B}_c(\Xm)$ and $\lambda \in [0,A - 2  m^\phi \lbc ]$.
    Then 
    \begin{align}&
    	  \int_{\Ka {\Xd} }\exp \left\{ \lambda \lnormd {\eta} {Q_k}^2 \right\}
			    \pi _{{\Delta} }(\mathrm{d}\eta |\xi )
	    \notag \\ \leq& 
	    \exp\bigg\{ \Upsilon_{{\Delta},\e} + \Big( \frac 1 {2} \e \betaD \|\phi\|_\infty C_\Delta  + \betaD  m^\phi \lbc \Big)
	    		\sum_{l \in  \mathcal{K}_{\mathcal{U}_{\Delta}}}
									\lnormd \xi {Q_l}^2 
						\bigg\}
	     < \infty
	     \label{B28c},
	   \end{align}
	   where	$\e > 0$ is arbitrary, $C_\Delta$ was defined in \refeq{eq314b} and
	   $$
				\Upsilon_{{\Delta},\e} := 
							 \betaD \|\phi\|_\infty \bigg(
								C_\Delta^2 
						  + \frac 1 {2 \e}  C_\Delta
						  	|\mathcal{K}_{\mathcal{U}_{\Delta}}|
						 \bigg) < \infty.
			$$
			\end{lemma}

    \begin{proof}
    By Lemma \ref{Bl2a} the integral in \refeq{B28c} does not exceed
    \begin{align}&
    	\frac 1 {Z_{{\Delta}}(\xi)}
					\int_{\K} \exp \bigg\{
 			      			-\left[   A - 2 m^\phi \lbc - \lambda\right] \lnormd {\eta_\Delta} {Q_k}^2
			 		\notag \\ & \qquad \qquad \qquad
			 		\qquad
			 				- \left[ A - 2  m^\phi \lbc \right]  \sum_{\substack{j \in {\mathcal{K}_{\Delta}} \\ j \neq k}} 
								\lnormd {\eta_\Delta} {Q_j }^2
								\bigg\}
					\PmeaKd \theta {\Delta}(d\eta_{\Delta})
				\notag \\ &\quad 
	    		\times	\exp \Big\{ 	 
	    		\betaD 
	    		 m^\phi \lbc 
	    		\sum_{l \in  \mathcal{K}_{\mathcal{U}_{\Delta}}}  
									\lnormd {\xi_{\Delta^c}} {Q_l }^2  
									\Big\}
		\label{B28c2}.
    \end{align}
    Lemma \ref{Bl2b} (cf. Eq. \refeq{Beq10ab}) yields the claim, where we note that
			\begin{align}
				&
					\sum_{j \in \mathcal{K}_{\mathcal{U}_{\Delta}}} \lnormd \xi {Q_j}
				\leq
					\frac 1 {2 \e} |\mathcal{K}_{\mathcal{U}_{\Delta}}| 
					+ \frac \e 2 \sum_{j \in \mathcal{K}_{\mathcal{U}_{\Delta}}} 
							\lnormd \xi {Q_j}^2
				\label{B28d}.
			\end{align}
			\end{proof} 
	In Section \ref{Bss5}, we will improve this result by showing that the right-hand side of \refeq{B28c} can be bounded uniformly as ${\Delta} \nearrow \Xm$. 
	
	\subsection{Weak dependence on boundary conditions}
\label{Bss5}
In this section we prepare some technical estimates on the local specification kernels, which will be crucial to prove the existence of an associated Gibbs measure $\mu \in \tGibK \pot \lambda \Xm$. To this end, we use an inductive scheme that is based on the consistency property (\ref{B16a}). We start by deducing the following bound in the  elementary   cubes ${} Q_k$, $k \in \ZZ$.

	\begin{lemma}
	\label{Bl4} 
	Fix $k\in \ZZ$, $\xi \in \Ksad  $, ${\Delta} \in \mathcal{B}_c(\Xm)$ and $\lambda \in [0, A -  m^\phi \lbc ]$.
    Then
	\begin{align}                    
		 \int_{\K} \exp \left\{ \lambda \lnormd {\eta} {Q_k}^2 \right\} 
		 		    \pi _{k }(\mathrm{d}\eta |\xi )
			\leq&
			\exp \bigg\{ \Upsilon_{\e} + \Big( \lbc +  \e C_\phi \Big) \sum_{j \in \partial^{\phi} k}
								\lnormd \xi {Q_j}^2
			     \bigg\} 
		\label{B30},
	\end{align}
	where $\e > 0$ is arbitrary and 
			\begin{eqnarray}
				C_\phi&:=&  \theta  \betaD g^d \|\phi\|_\infty 				
			\notag, \\
			\Upsilon_{\e} 
				&:=&
				  C_\phi \left( 4 \theta g^d  
				  			+ \frac { m^{\phi} } {\e}
				  			\right).
			\label{B9a4}	
			\end{eqnarray}
	\end{lemma}

\begin{proof}	
    	The claim follows immediately by \refeq{37b} and by the proof of Lemma \ref{Bl2}. 
\end{proof}

\begin{rem}\label{remB4}
Estimate \refeq{B30} states the so-called \emph{weak dependence} on boundary conditions. Analytically this means, by choosing $\lambda \in (m^\phi \lbc , A -  m^\phi \lbc)$, that 
$$
	\left( \lbc + \e C_\phi  \right) m^{\phi} < \lambda < A -  m^\phi \lbc ,
$$
which is always possible by \refeq{Bex20a} for small enough $\e > 0$. Applying Jensen's inequality to both sides in \refeq{B30}, we get a Dobrushin's type estimate (cf. Eq. (4.9) in \cite{pre76})
\begin{align}
	\int_{\K} \eta(Q_k)^2 \pi_k(d\eta|\xi) 
	\leq
		\frac 1 \lambda \left\{ \Upsilon_\e + \left( \lbc + \e C_\phi \right) \sum_{j \in \partial^\phi k} \xi (Q_k)^2 \right\}.
		\label{45b}
\end{align}
However, we cannot make a straightforward use of the Dobrushin existence \cite[Theorem 1]{pre76} for several technical reasons caused, e.g., by lacking a Polish metric on $\K$ and by the discontinuity of the specification kernels $\pi_\Delta(d \eta|\xi)$ subject to the boundary condition $\xi \in \K$.
\end{rem}

  	Consider now arbitrary large domains ${\Delta} _{\mathcal{K}} = \bigsqcup_{k\in \mathcal{K}}{{} Q_k} \in	\mathcal{Q}_{c}(\Xm)$ indexed by $\mathcal{K}\Subset \ZZ$. 
	Note that ${\Delta} _{\mathcal{K}} \nearrow \Xm$ as $\mathcal{K}\nearrow \ZZ$. Using the estimate (\ref{B30}) and the consistency property (\ref{B16a}), our next step will be to get similar moment estimates for all specification kernels $\pi _{\mathcal{K}}(\mathrm{d}\eta |\xi ):=\pi _{{\Delta} _{\mathcal{K}}}(\mathrm{d}\eta |\xi )$.
			
		\begin{prop}\label{Bp3} 
			Let $0 \leq \lambda \leq A -  m^\phi \lbc$. Then there exists $\mathcal{C}_\lambda < \infty$ such that for all $k\in \ZZ$ and $\xi \in \tKad \lomagsy \Xm $
			\begin{equation}
					\limsup_{\mathcal{K}\nearrow \ZZ}
						\thinspace
						\int_{\K} 
							\exp \left\{ \lambda \lnormd {\eta} {Q_k}^{2}\right\} 
						\pi _{\mathcal{K}}(d\eta |\xi )
				\leq
					\mathcal{C}_{\lambda}. 
			\label{B35}
			\end{equation}
Moreover, for each $\alpha >0$ one finds  proper $\nu _{\alpha }>0$ and $\mathcal{C}_\alpha < \infty$ such that
			\begin{equation}
				\limsup_{\mathcal{K}\nearrow \ZZ}
					\int_{\K} 
						\exp \left\{ \nu _{\alpha } \anorm \alpha \eta^{2}\right\} 
					 \pi _{\mathcal{K}}( d \eta |\xi )
				\leq 
					\mathcal{C}_\alpha.  \label{B36b}
			\end{equation}
		
	\end{prop}
	
\begin{proof}
Without loss of generality, we assume that $\lambda \in ( m^\phi \lbc , A -  m^\phi \lbc]$. Let us define 
\begin{align}
	0 \leq n_{k}(\mathcal{K}|\xi )
	:=&
		\log \left\{ 
			\int_{\K} \exp \left\{ \lambda \lnormd {\eta} {Q_k}^{2} \right\} 
			\pi_{\mathcal{K}}(d\eta |\xi )\right\}, \quad k \in \ZZ
	\label{B32},
\end{align}
which are finite by Lemma \ref{Bl2}. In particular,
\begin{equation*}
n_{k}(\mathcal{K}|\xi ):=\Phi (\xi _{k})\text{ \ \ if \ }k\notin \mathcal{K}.
\end{equation*}%
Next, we will find a global bound for the whole sequence 
$\left( n_{k}(\mathcal{K}|\xi )\right) _{k\in \ZZ}$, which then
implies the required estimates on each of its components.\\

Integrating both sides of (\ref{B30}) with respect to $\pi _{\mathcal{K}}(%
\mathrm{d}\eta |\xi )$ with an arbitrary $\xi \in \tKad \lomagsy \Xm$ and taking into account the consistency property (\ref{B16a}),
we arrive at the following estimate for any $k\in \mathcal{K}$%
\begin{align}
	&n_{k}(\mathcal{K}|\xi ) 	
	\leq 
		\Upsilon_{\e}
		+ \log \left\{  
				\int_{\K} \exp \left[ 
				 	 \bigg( {\betaD  \lbc } 
					+ \e C_\phi \bigg) \sum_{j \in \partial^{\phi} k}
							\lnormd \xi {Q_j}^2
		    \right] 
				\pi _{\mathcal{K}}(d\eta |\xi )
				\right\}
		\notag \\& = 			
			\Upsilon_{\e}
			+	\left[ 
				 	\bigg( {\betaD  \lbc } + \e C_\phi \bigg) \sum_{j \in \mathcal{K}^c \cap \partial^{\phi} k}
							\lnormd \xi {Q_j}^2
		    \right] 			
		\notag \\ &+ \log \left\{  
				\int_{\K} \exp \left[ 
				 	\bigg( {\betaD  \lbc }
					+ \e C_\phi \bigg) \sum_{j \in \mathcal{K} \cap \partial^{\phi} k}
							\lnormd \xi {Q_j}^2
		    \right] 
				\pi _{\mathcal{K}}(d\eta |\xi )
				\right\}
		\label{B33a}.
	\end{align}
We will apply the multiple H\"{o}lder inequality
\begin{equation}
\mu \left( \prod\nolimits_{j=1}^{K}f_{j}^{t_j}\right) \leq
\prod\nolimits_{j=1}^{K}\mu ^{t_j}(f_{j}),\text{ \ \ \ }\mu (f_{j}):=\int
f_{j}\mathrm{d}\mu ,  \label{B38}
\end{equation}%
valid for any probability measure $\mu $, nonnegative functions $f_{j}$, and $t_j\geq 0$ such that $\sum_{j=1}^{K}t_j\leq 1,$ $K \in \NN$.
Choose $\delta,\e>0$ such that      
\begin{align}
       	0 <  B_{\e} :=
			 \left( \lbc + \e  C_\phi \right) m^{\phi}
		<
			  \delta \lambda
		< \lambda
		\leq
			\lambda_0
					:=
					A -  m^\phi \lbc
	\label{B28bb},
\end{align}      
where $C_\phi$ is defined by \refeq{B9a4}. In our context $f_{j}:=\exp \left\{ \lambda \lnorm {\eta _{j}}^2\right\}$ and 				
$$
	t_j := 1/ \lambda \left( \lbc + \e C_\phi \right),
	\qquad \text{for $j \in \mathcal{K} \cap \partial^{\phi} k$.}
$$
Using this setting for \refeq{B38},
we deduce that the last summand in \refeq{B33a} is dominated by
\begin{align}
	& \sum_{\substack{j\in \mathcal{K}\cap \partial^{\phi} k}}
		\log \left\{
		\int_{\K} \exp \left( \lambda \lnormd \eta {Q_j}  \right)
			\pi_{\mathcal{K}}( d \eta|\xi) \right\}^{t_j}
	\notag \\ & \qquad =
		1/ \lambda \left( \lbc + \e C_\phi \right)
		\sum_{j\in \mathcal{K}\cap 
					\partial^{\phi} k}
			n_j(\mathcal{K}|\xi)
\label{B37d}.
\end{align}

		Let $\mathcal{K}\Subset \ZZ$ contain a fixed point $k_{0}\in \ZZ$. Let  $\vartheta := R/g + \sqrt{d}$ be such that $|j-k_0| \leq \vartheta $ for all $j\in \partial^{\phi} k_0$. Let us pick an $\alpha >0$ small enough, so that $B_{\e} e^{\alpha \vartheta} < \lambda$. We take the sum over $k \in \mathcal{K}$ of the terms $n_{k_0}(\mathcal{K}|\xi)$ with the weights $\exp \{ - \alpha |k|\}$. After rearranging, by Eqs. \refeq{B33a} and \refeq{B37d} we get for each $k \in \mathcal{K}$ 
	\begin{align}
			n_{k_{0}}(\mathcal{K}|\xi )
		\leq&
			 \sum_{k\in \mathcal{K}}\left[ n_{k}(\mathcal{K}|\xi )
			 				\exp \{-\alpha | k | \}\right]  
		\label{B39} \\
		\leq&
			 \left[ 1- \frac{ B_{\e}}
			 			{\lambda }
			 		e^{\alpha \vartheta }
			 \right] ^{-1}
			 \left[ \Upsilon_{\e}
			 		+ B_{\e} 
			 			 e^{\alpha \vartheta }
			 			||\xi _{\mathcal{K}^{c}}||_{\alpha }^{2}
			 	\right].
	  \notag
	\end{align}
Note that 
\begin{align}
		&\left[ 1- \frac{B_{\e} e^{\alpha \vartheta }}
			 			{\lambda }			 		
			 \right] ^{-1}
		=
			1 + \frac{ B_{\e} e^{\alpha \vartheta }} 
			 			{\lambda - B_{\e} e^{\alpha \vartheta }}
		\leq
			\frac 1 {1 - \delta e^{\alpha \vartheta }} 
		\notag,
\end{align}
where we used \refeq{B28bb}. Plugging this back into \refeq{B39}, we get
\begin{align}&
		n_{k_{0}}(\mathcal{K}|\xi )
		\leq
			 \sum_{k\in \mathcal{K}}\left[ n_{k}(\mathcal{K}|\xi )
			 				\exp \{-\alpha | k | \}\right] 
		\notag \\ \leq&
			\left[ \Upsilon_{\e}
			 		+ B_{\e}  
			 			 e^{\alpha \vartheta}
			 			\anorm \alpha {\xi _{\mathcal{K}^{c}}}^{2}
			 	\right]
			 	\left( \frac{1 } {1 - \delta e^{\alpha \vartheta }} \right)
		\label{B39b}.
\end{align}					
Since $\anorm \alpha {\xi _{\mathcal{K}^{c}}}$ tends to zero as $\mathcal{K}\nearrow \ZZ$, we obtain for each $k_{0}\in \ZZ$%
\begin{align}
		&\limsup_{\mathcal{K}\nearrow {\ZZ}}
			\text{\thinspace }
				\sum_{k\in \mathcal{K}}
						\left[ n_{k}(\mathcal{K}|\xi )
								\exp \{-\alpha | k | \}\right]
		\leq 
			\Upsilon_{\e}
				\left( \frac{1} {1 - \delta e^{\alpha \vartheta }} \right)
		\label{B40a}, 
\end{align}%
and thus, by letting $\alpha \searrow 0$,
we complete the proof of \refeq{B35}
\begin{align}
		\limsup_{\mathcal{K}\nearrow \ZZ}
			\text{\thinspace }
						n_{k_0}(\mathcal{K}|\xi )
		\leq &
			\frac 1 {1 - \delta} \Upsilon_{\e}
			=: 
			\log \mathcal{C}_{\lambda}
		\label{B40abb}.
\end{align}
So, we have for each $\lambda$ fulfilling \refeq{B28bb} (and hence also for all $\lambda \leq \lambda_0 :=  A - m^\phi \lbc $)
\begin{align}
	&\limsup_{\mathcal{K}\nearrow \ZZ}
			\int_{\K} \exp 
				\left\{ \lambda \lnormd {\eta} {Q_k}^{2} \right\} 
			\pi _{\mathcal{K}}(d \eta |\xi )  
	\leq 
\mathcal{C}_{\lambda}.  
\label{B40d} 
\end{align}
By the H\"{o}lder inequality \refeq{B38} we see that Eq. \refeq{B36b} holds with
\begin{align}
	\nu _{\alpha}
	:=&
		\lambda \bigg[\sum_{k\in \ZZ} \exp \{- \alpha | k | \} \bigg]^{-1}
		\label{B10b18}.
\end{align}
Indeed, we have 
\begin{align}
	&	\int_{\K} \exp 
				\bigg\{ \nu_{\alpha} \sum_{k \in \ZZ} \lnormd {\eta} {Q_k}^{2} 
					e^{-\alpha | k |} \bigg\} 
			\pi_{\mathcal{K}}(d \eta |\xi )
	\notag 	\\	 \leq&
		 \bigg(
				\exp\Big\{ \sum_{k \in \mathcal{K}} n_k(\mathcal{K}|\xi) 
										\exp\{-\alpha | k |\} \Big\}
			\bigg)^{ \frac {v_{\alpha}} {\lambda_0} 
								\exp\{-\alpha | k |\}}
			e^{ \anorm \alpha {\xi_{\mathcal K^c}}^2 }
	\notag.
	\end{align}
	Hence, using also \refeq{B40a}, we deduce
	\begin{align}&
		\limsup_{\mathcal K \nearrow \ZZ} \int_{\K} \exp \left\{ \nu_\alpha \anorm \eta \alpha^2 \right\}
		\pi_{\mathcal K} (d\eta|\xi)
		\notag \\ & \leq 
			\exp \left\{\Upsilon_\e \left( \frac 1 {1 - \delta e^{\alpha \vartheta} } \right) \right\}
			=: C_\alpha 
		\label{418b},
		\end{align}
		which completes the proof of \refeq{B36b}.
\end{proof}

	 An import corollary of the above proposition claims uniform bounds for local Gibbs states:\label{Bss9c3}
		\begin{cor}\label{Bcor9c4a}
			Let Assumption \emph{\textbf{$\mathbf{(\pot)}$}} hold. Then for all $\Delta \in \mathcal{Q}_c(\Xd)$ and $N \in \NN$ there exists $\mathcal{C} (\Delta,N) < \infty$ such that  
			\begin{align}
				\limsup_{\substack{{\tilde \Delta} \nearrow \Xm\\ {\tilde \Delta} \in \mathcal{Q}_c (\Xd)}}
						\int_{ \Ka {\Xd}} \lnormd \eta \Delta^N 
						\pi_{\tilde \Delta}(d\eta|\xi) 
						\leq& \mathcal{C} (\Delta,N) < \infty,
				\notag 
			\end{align}
where $\mathcal{C}(\Delta,N)$ can be chosen uniformly for all $\xi \in \tKad \floma \Xm$.
	\end{cor}

	\subsection{Local equicontinuity}\label{Bsss4c4a}
	A Gibbs measure $\mu \in \GibK \pot \lambda \Xm$ will be constructed as a cluster point of the net of specification kernels $\{ \pi_\Delta\}_{\Delta \in \B_c(\Xd)}$. To this end, an important step is to establish the equicontinuity of the local specification, which yields the existence of limit points in a proper topology. 
	
    \begin{defi}\label{Bdef9c4Ea}
    On the space of all probability measures ${\mathcal{P}}(\K)$ we introduce the topology of \emph{$\mathcal{Q}$-local convergence}. This topology, which we denote by $\mathcal{T}_{\mathcal{Q}} $, is
    defined as the coarsest topology making the maps ${\mathcal{P}}(\K) \ni \mu
    \mapsto \mu (B)$ continuous for all sets $B \in \B_{\mathcal{Q}} (\K)$.
    Here,
     \begin{align}
			\B_{\mathcal{Q}} ({\K}) := \bigcup_{{\Delta} \in \mathcal{Q}_c (\Xm)}
				\B_{\Delta} ({\K})
		\label{B80fb}
		\end{align}
    denotes the \emph{algebra of all local events} associated with the partition cubes,  where
      $\B_\Delta (\K) := \mathbbm{P}^{-1}_\Delta \B(\Ka \Delta)$ and the canonical projections $\mathbbm{P}_\Delta$ were defined by \refeq{eq24b}.
   \end{defi} 

Note that the topology of local convergence is \emph{not metrizable} (see  \cite[p. 57]{geo79}. So, to describe the corresponding convergence one has to consider nets instead of sequences. Let us recall that a subset of measures from $\mathcal{P}(\Kad)$ is relative compact iff each of its net has a cluster point in $\mathcal{P}(\Kad)$; furthermore, every cluster point can be obtained as a limit of a certain subnet. A sufficient condition fo the existence of cluster points is the so-called \emph{equicontinuity} property. 

More precisely, we modify \cite[Definition 4.6]{geo88} to fit our setting: 
    \begin{defi} \label{Bdef9c4b}
       Fix $\xi \in \tKad \lomagsy \Xm$. The net $
       	\{\pi_{\Delta}(d\eta |\xi )| {\Delta} \in \mathcal{Q}_c (\Xm)\}
       $
      is called \emph{$\mathcal{Q}$-locally equicontinuous} 
      if for all 
      ${\tilde {\Delta}} \in \mathcal{Q}_c(\Xm)$ and for each sequence $
      \{B_N\}_{N \in \mathbb{N}}\subset \mathcal{B}_{\tilde {\Delta}} ({\K})$ 
      with $B_{N}\downarrow \varnothing $%
      \begin{equation}
			  \lim_{N\rightarrow \infty }\ 
			  \limsup_{\substack{ {\Delta} \nearrow \Xm\\ {\Delta} \in \mathcal{Q}_c (\Xm)}}
				  \pi _{{\Delta} }(B_{N}|\xi )=0.  
				 \label{B80f}
      \end{equation}
    \end{defi}
A crucial issue (resulting from \cite[Proposition 4.9]{geo88}) is that each $\mathcal{Q}$-local equicontinuous net has at least one $\tau_{\mathcal{Q}}$-cluster point in $\mathcal{P}(\K)$ (cf. \cite[Proposition 5.3]{pre05}). 
As indicated above, our first aim is to prove the $\mathcal{Q}$-local equicontinuity for the specification kernels.
    \begin{prop} \label{Bp9c4c}
Let Assumption \emph{\textbf{$\mathbf{(\pot)}$}} hold.
      Then for each fixed $\xi \in \tKad \lomagsy \Xm$ the net
    $\{\pi_{{\Delta} }(\mathrm{d}\eta |\xi )|$ ${\Delta} \in \mathcal{Q}_c (\Xm)\}$
    is locally equicontinuous.
    \end{prop}

\begin{cor}\label{cor46}
Let Assumption \emph{\textbf{$\mathbf{(\pot)}$}} hold and fix  some order generating sequence $\{\Delta_N \}_{N \in \NN} \subset \mathcal{Q}_c (\Xm)$. Then, for each boundary condition $\xi \in \tKad \lomagsy \Xm$,  (a subsequence of) $\{ \pi_{\Delta_N}(\cdot|\xi) \}_{N \in \NN}$ converges $\mathcal{Q}$-locally to a probability measure $\mu \in \mathcal{P}(\Kad)$, that means for all $B\in \B_{\mathcal{Q}}(\K)$
        \begin{equation}
	      \pi _{{{\Delta}}_N}(B|\xi )\rightarrow \mu (B)\text{ \ \ as }%
        N\rightarrow \infty  
        \notag.
        \end{equation}
\end{cor}
\begin{proof}
The claim follows by combining Proposition \ref{Bp9c4c} with \cite[Propositions 4.9 and 4.15]{geo88} and \cite[Theorem V.3.2]{par67}.
\end{proof}

\begin{proof}[Proof of Proposition \ref{Bp9c4c}]   
    We will split the set $B_N$ and then use the support property 
    (cf. Corollary \ref{Bcor9c4a}), the consistency property 
    \refeq{B16a} and the lower bound for a partition 
    function (cf. Lemma \ref{Bl2b}) to estimate the 
    two summands. (The basic idea is given by an adaption 
    of the arguments used for proving \cite[ Theorem 4.12 
    and Corollary 4.13]{geo88} to the configuration space 
    setting, cf. also \cite{koparo10}.)\\

Let us fix any ${\tilde \Delta} \in \mathcal{Q}_c(\Xm)$, and let $\{B_{N}\}_{N\in \mathbb{N}}$ be any sequence of sets from $\mathcal{B}_{\tilde \Delta}({\K})$ such that $B_{N}\downarrow \varnothing $ as $N\rightarrow \infty .$ Consider the following Borel subsets in $\Ka \Xd$ consisting of those measures $\eta$ whose local masses over $\mathcal{U}$ are bounded by a given $T>0$,
        \begin{equation}
	         \Kas \lbrack \mathcal{U},T \rbrack:=
	        \left\{ \eta \in {\K}\ \left\vert \
	                    \lomagIf {\mathcal{U}} \eta
	                 \leq T\right. \right\} ,\text{ \ \ }T>0,
	                    \notag
        \end{equation}
        where we define, using \textbf{(FR)},  
        \begin{align}
        	\mathcal{U}:= 
        	\bigsqcup_j \left\{ Q_j \big| \exists \thinspace x  \in Q_j, \thinspace \exists y \in \tilde \Delta: |x-y| \leq R \right\} 
        	 \in \mathcal{Q}_c (\Xm)
        	\label{BB10b2b}.
        \end{align}

        For each $\xi \in  \Ksad  $ and ${{\Delta}} \in \mathcal{Q}_c (\Xm)$ which contains ${\tilde \Delta}$, we have by the consistency property 
        \begin{align}
                \pi _{{{\Delta}} }(B_{N}|\xi )
            & =  \pi _{{{\Delta}} }(B_{N}\cap \lbrack \Ka {\mathcal{U},T}]^{c}|\xi )
                + \int_{\K} \pi_{{{{\tilde \Delta}}}} (B_N \cap  \Kas \lbrack {\mathcal{U},T} \rbrack | \eta) \pi_{{\Delta}} (d\eta|\xi)
        \notag \\ &= 
                \pi _{{{\Delta}} }(B_{N}\cap \lbrack \Ka {\mathcal{U},T} \rbrack^{c}|\xi )
           \notag \\ &\quad
            +
                \int_{\K}   \frac 1 {Z_{{{{\tilde \Delta}}}} (\eta)} \int_{\Ka {{\tilde \Delta}}} \mathbbm{1}
                _{B_{N}\cap \Kas \lbrack {\mathcal{U},T}\rbrack }(\rho _{{{{\tilde \Delta}}}} \cup
                \eta_{{\tilde \Delta}^c})
            \notag \\ & \qquad
                    \exp \left\{ -\betaD H_{{{{\tilde \Delta}}}}(\rho_{{{{\tilde \Delta}}}} | \eta )\right\}
                   \PmeaKd \theta {{\tilde \Delta}}  
								(\mathrm{d}\rho_{{{{\tilde \Delta}}}})
                    \pi_{{\Delta}}(\mathrm{d}\eta |\xi ).
        \label{B90a}
            \end{align}
	  By Chebyshev's inequality and Corollary \ref{Bcor9c4a} we get  
    \begin{align}
      & \limsup_{\substack{{{\Delta}} \nearrow \Xd \\ {{\Delta}} \in \mathcal{Q}_c(\Xd)}} \pi_{{\Delta}} (\Kas \lbrack \mathcal{U},T \rbrack ^c) 
      = 
      \limsup_{\substack{{{\Delta}} \nearrow \Xd \\ {{\Delta}} \in \mathcal{Q}_c(\Xd)}} 
      \pi_{{\Delta}} \left( \left\{\eta \in {\K} :
          \left. \lomagIf {\mathcal{U}} \eta > T  \right\} \right| \xi \right)
      \notag \\  \leq&
      \limsup_{\substack{{{\Delta}} \nearrow \Xd \\ {{\Delta}} \in \mathcal{Q}_c(\Xd)}} 
          \int_{{\K}} \frac {\lomagIf {\mathcal{U}} \eta} {T} 
          	\pi_{{\Delta}} (d\eta |\xi)
    \leq 
    \frac 1 T \mathcal C (\mathcal{U},1) 
		< \infty      
     \label{B90b},
     \end{align}
     which vanishes as $T \nearrow \infty$.
    
    On the other hand, for all $\eta \in {\K}$ and ${{\Delta}} \in
    \mathcal{Q}_{c}( \Xm)$ containing ${\tilde \Delta}$, we estimate the 
    outer integrand (using \refeq{Beq10ab}, \refeq{BeqB9aa} and the elementary inequality $|ab| \leq 1/2 (a^2 + b^2)$ for $a,b \in \RR$) as follows
    \begin{align}
    	 &\frac 1 {Z_{{{\tilde \Delta}}} (\eta)}
				\int_{ \Ka {{{\tilde \Delta}} }}
					\mathbbm{1}_{B_N\cap  \Kas \lbrack {\mathcal{U},T} \rbrack} 
								(\rho_{\tilde \Delta}\cup\eta_{{\tilde \Delta}^{c}})
           \exp \left\{ -\betaD H_{{{\tilde \Delta}}}
								(\rho_{\tilde \Delta} |\eta )\right\}
                \PmeaK {{{\tilde \Delta}}} (d \rho_{\tilde \Delta})
     \notag \\ \leq &
       	\exp \bigg\{ \betaD \|\phi\|_\infty \Big( 
       			C_{{\tilde \Delta}}^2 + C_{{\tilde \Delta}} \cdot \lomagIf {\mathcal{U}} {\eta_{\tilde \Delta^c}}  \Big)
	       	\bigg\}
	   		\int_{ \Ka {\tilde \Delta}}
					\mathbbm{1}_{B_N\cap  \Kas \lbrack {\mathcal{U},T} \rbrack} 
								(\rho_{\tilde \Delta}\cup\eta_{{\tilde \Delta}^c})
				\notag \\ &
						\exp \left\{ \betaD \|\phi\|_\infty 
								\left( \lomagIf {{\tilde \Delta}} \rho ^2 + 2 \lomagIf {{\tilde \Delta}} \rho 
					\lomagIf {\mathcal{U}} {\eta_{\tilde \Delta^c}} \right)
								\right\}
				\PmeaKd \theta {{\tilde \Delta}}  (d\rho)				
	   	\notag \\ \leq&
	   		\exp \bigg\{ \|\phi\|_\infty \Big( 
       			C_{{\tilde \Delta}}^2 + 1/2 T^2 + 1/2 C_{{\tilde \Delta}}^2 \Big)
	       	\bigg\}
	      e^{3\betaD \|\phi\|_\infty T^2} 
	   		\PmeaKd \theta {{\tilde \Delta}}  \left(B_N \cap  \Kas \lbrack {\mathcal{U},T} \rbrack\right)
	  \notag \\ \leq&
	   		C e^{4\betaD \|\phi\|_\infty T^2} 
	   		\PmeaKd \theta {{\tilde \Delta}}  \left(B_N \right)
	   		< \infty,
	  \notag
		\end{align}
		  where $C$ is an appropriate constant. Summing up, we get 
		\begin{align}&%
	    \pi _{\Lambda }(B_N \cap \Ka {\mathcal{U},T} |\xi ) 
             \leq 
	    C e^{4\betaD \|\phi\|_\infty T^2} 
	   		\PmeaKd \theta {{\tilde \Delta}} \left(B_N \right)
	  \label{B90c},
    \end{align}%
    which tends to zero for $B_N \searrow \varnothing$.
   	Plugging (\ref{B90b}) and (\ref{B90c}) back into (\ref{B90a}), we get the equicontinuity of the family
    $\{\pi _{\Lambda}(\mathrm{d}\eta |\xi )|$
    $\Lambda \in \mathcal{Q}_c (\Xm) \}$ required in Eq. \refeq{B80f}.
\end{proof}

		   \subsection{Existence of Gibbs measures}
   \label{Bss4c5b}
   Now we are in position to deduce the main result. Namely, we show that each limit point that we obtained by the local equicontinuity proved above is indeed a Gibbs measure.
\label{Bp4c5b2}

\begin{thm} \label{Bthm9c4E}
	Let $\pot: \Xm \times \Xm \rightarrow \RR$ be such that Assumption  \emph{\textbf{$\mathbf{(\pot)}$}} holds. Then there exists a Gibbs measure $\mu$ (at least one) corresponding to the potential $\pot$ and the Gamma measure $\PmeaK {\theta}$, which is supported by $\tKad \lomagsy \Xm$.\label{Bthm3a2f}
	Therefore, 
	$$
		 \tGibK {\pot} \theta \Xm \neq \varnothing.
	$$
	Furthermore, the set $\tGibK {\pot} \theta \Xm$ is relatively compact in the topology $\mathcal{T}_{\mathcal{Q}}$.
\end{thm}

\begin{proof}
Fix  some order generating sequence ${\Delta}_N \nearrow \Xm$, ${\Delta}_N \in \mathcal{Q}_c (\Xm)$, and a boundary condition $\xi \in \tKad \lomagsy \Xm$.
By Corollary \ref{cor46}, (a subsequence of) $\{ \pi_{\Delta_N}(\cdot|\xi) \}_{N \in \NN}$ converges $\mathcal{Q}$-locally to a probability measure $\mu \in \mathcal{P}(\Kad)$.

The limit point $\mu$ is surely Gibbs. Indeed, the DLR property can be derived easily: Fix ${\Delta} \in \mathcal{Q}_c(\Xm)$ and $B\in \B_{\mathcal{Q}}(\K)$ arbitrarily. By $\mathbf{(FR)}$ and \refeq{BB10b2b} we can justify  the following equations to get the $({\bf DLR})$ one.
   \begin{align}
       		\int_\K \pi_{{\Delta}} (B | \eta) \mu (d \eta)
       \oset {1.} = & 
	       \lim_{N \rightarrow \infty}
	       			\int_\K \pi_{{\Delta}} ( B | \eta_{\mathcal{U}_{{\Delta}}} ) \pi _{{{\Delta}}_N}(d \eta |\xi )
	    \notag \\ \oset {2.} = &
	       \lim_{N \rightarrow \infty}
	       \int_\K \pi_{{\Delta}} ( B | \eta_{\mathcal{U}_{{\Delta}}} ) 
	       \pi _{{{\Delta}}_N}(d \eta |\xi )
	   \notag \\ \oset {3.} =
               & \lim_{N \rightarrow \infty} \int_\K \pi_{{\Delta}} ( B | \eta) 
               \pi _{{{\Delta}}_N}(d \eta |\xi )
              \notag \\ \oset {4.} =&
	       \lim_{N \rightarrow \infty}
	         \pi _{{{\Delta}}_N}(B|\xi )
	      \notag \\ \oset{5.} = &
	        \mu (B).
	  \label{B9c4E2}
            \end{align}
            The first and third equality follow by the choice of             $\mathcal{U}_{\Delta}$, the second and fifth one by the definition of $\mu$ and the fourth one by the consistency of the local specifications (cf. \refeq{B16a}).\\
            
  The relative compactness of $\tGibK {\pot} \theta \Xm$ follows analogously. Indeed, similarly as above one shows that every net in $\GibK \pot \theta \Xm $ has a $\mathcal{T}_{\mathcal{Q}}$-cluster point in ${\mathcal{P}}(\K)$, which is (reasoning as in Eq. \refeq{B9c4E2}) a Gibbs measure. 
\end{proof}
 	
	\subsection{Moment estimates for Gibbs measures}\label{Bss10b4}
	\begin{thm}\label{Bthm4c8}
		Let Assumption \emph{\textbf{$\mathbf{(\pot)}$}} hold. Then, for each $\lambda \in \left[0, A -  m^\phi \lbc \right]$ there exists an (explicitely computable) $\mathcal{C}_\lambda < \infty$ such that uniformly for all $\mu \in \tGibK \pot \theta \Xm$ 
		\begin{align}
			\sup_{k \in \mathbbm{Z}^d} \int_{\K} \exp \left\{ \lambda \eta(Q_k)^2 \right\} \mu(d\eta) \leq \mathcal{C}_\lambda.
		\label{426b}
		\end{align}
		Furthermore, for each $\alpha >0$ one finds a certain $\nu _{\alpha }>0$ such that for all $\mu \in \tGibK {\pot} \theta \Xm$
			\begin{equation}
					\int_\K
						\exp \left\{ \nu _{\alpha } \anorm \alpha \eta^{2}\right\}
					 \mu( d \eta)
				\leq
					\mathcal{C}_\alpha.  \label{B36bx}
			\end{equation}
		Here, $\mathcal{C}_\lambda$ and $\mathcal{C}_\alpha$ are the same constants as in \refeq{B40d} resp. \refeq{418b}.
             \end{thm}

             \begin{proof}
             Let us only prove \refeq{B36bx}; the proof of \refeq{426b} is similar. 
         Using Beppo Levi, we have
\begin{align}
      	&\int_ {\K}  \exp \left\{ \nu _{\alpha } \anorm \alpha \eta^{2}\right\} \mu(d\eta)
            =
             \lim_{M \nearrow \infty}
	    \int_ {}   \exp \left\{ \nu _{\alpha } \anorm \alpha {\eta }^{2} \wedge M \right\}  \mu(d\eta)
	   \label{426c}. \\
	   \intertext{By the DLR equation  the right-hand side equals }
	    &
	   \lim_{M \nearrow \infty}
	   \lim_{\substack{\Delta \nearrow \Xm \\ \Delta \in {\mathcal Q}_c(\Xm)}}
					\int_ {\Ka  \Xm}
						\int_ {\Ka  \Xm}
						 \exp \left\{ \nu _{\alpha } \anorm \alpha {\eta }^{2} \wedge M \right\}
						\pi_{\Delta}(d\eta | \xi )
					\mu(d \xi)
	\notag. \\
		\intertext{By Lebesgue's dominated convergence theorem, the latter equals}
		&
             \lim_{M \nearrow \infty}
	   		\int_ {\Ka  \Xm}
	   		\bigg(  \lim_{\substack{\Delta \nearrow \Xm
	   						\\ \Delta \in {\mathcal Q}_c(\Xm)}}
					\int_ {\Ka  \Xm}
					 \exp \left\{ \nu _{\alpha } \anorm \alpha {\eta_{\Delta} }^{2} \wedge M \right\}
						\pi_{\Delta}(d\eta_{\Delta} | \xi )
               \bigg)
				\mu(d\xi)
     \notag \\ &\leq 
              \lim_{M \nearrow \infty}
	   		\int_ {\Ka  \Xm}
	   		\bigg(  \lim_{\substack{\Delta \nearrow \Xm
	   						\\ \Delta \in {\mathcal Q}_c(\Xm)}}
					\int_ {\Ka  \Xm}
					 \exp \left\{ \nu _{\alpha } \anorm \alpha {\eta_{\Delta} }^{2} \right\}
						\pi_{\Delta}(d\eta_{\Delta} | \xi )
               \bigg)
				\mu(d\xi)
		\notag.
   \end{align} 
   Since $\mu(\tKad \floma \Xm)= 1$, we may apply the uniform bound proven in Proposition \ref{Bp3} (cf.  \refeq{B36b}). Thus, the right-hand side in \refeq{426c} is dominated by $\mathcal C_\alpha$, which was to be shown. 
\end{proof}

\begin{cor}\label{cor4c9}
Let Assumption \emph{\textbf{$\mathbf{(\pot)}$}} be fulfilled. For each ${\Delta} \in \mathcal{Q}_c(\Xm)$ and $\lambda > 0$, there exists $C_\lambda({\Delta})>0$ such that for all $\mu \in  \tGibK {\pot} \theta \Xm$  
			\begin{align}
				\int_{\K} e^{\lambda \lnormd \eta \Delta} \mu(d \eta) < C_\lambda({\Delta}) 
				\label{B11a2b1}.
			\end{align}
\end{cor}
One important consequences of the \emph{à-priori} bound \refeq{426b} is that  any $\mu \in \tGibK {\pot} \theta \Xm$ is \emph{à-posteori} carried by the (much smaller) Borel subset 

\begin{align}
	{\mathbbm{K}}^s (\Xd):= &\big\{ \eta \in \K \thinspace   
	\big| \thinspace \thinspace \forall b > d  \thinspace \thinspace \exists K_{\eta,b} \in \NN: 
		\notag \\ 
		  & \quad \eta(Q_k)^2 \leq b \log (1 + |k|) \qquad \text{ if $ k \in \ZZ$ with } |k| > K_{\eta,b} \big\}.
	\notag
\end{align}
     
\begin{prop}
	For all $\mu \in \tGibK \pot \theta \Xm$ 
	$$
		\mu \big( \mathbbm{K}^s (\Xd) \big)=1.
	$$
\end{prop}

\begin{proof}
Obviously, $\B(\K) \ni \Kas^s(\Xd) \subset \mathbbm{K}^t (\Xd)$. The compliment to $\Kas^s (\Xd)$ can be written as 
\begin{align}
	\left[ \Kas^s(\Xd) \right]^c = \bigcup_{b > d} \bigcap_{K \in \NN} \bigcup_{|k| \geq K} \left[ \Kas_k(b)\right]^c
	\label{414a},
\end{align}
where 
$$
	\Kas_k(b) := \left\{ \eta \in \K \big| \eta(Q_k)^2 \leq b \log ( 1 + |k| ) \right\}.
$$
By Chebyshev's inequality and Eq. \refeq{426b}
$$
	\mu \left( \left[ \Kas_k(b) \right]^c \right) \leq \mathcal{C}_\lambda (1 + |k|)^{-b};
$$
and therefore by \refeq{414a}
$$
	\mu\left( \left[ \Kas^s(\Xd) \right]^c \right) 
		\leq C_\lambda \lim_{b \searrow d} \ \lim_{K \rightarrow \infty}
			\sum_{|k| \geq K} (1 + |k|)^{-b} = 0.
$$
\end{proof}
	
\section{Some remarks and generalizations}\label{s4a}\label{s5}
\bitem
\item The existence result also holds for \emph{all} nonnegative (bounded) potentials $\phi$ with finite interaction range. For these potentials, we have for all $\Delta,{{\tilde \Delta}} \in \mathcal{Q}_c(\Xd)$ the following estimate (cf. \cite[Propositions 4.2.3 and 4.5.5]{hag11})
\begin{align}
	\int_{{\K}} \lomagIf {\Delta} \eta 
          	\pi_{{\tilde \Delta}} (d\eta |\xi)
    \leq  \theta m(\mathcal{U}_\Delta) + \xi_{{{\tilde \Delta}}^c}(\Delta)
  \notag.      
\end{align}
Plugging this into \refeq{B90b} and using that even $Z_{{{\tilde \Delta}}}(\gamma) \leq 1$, one obtains the local equicontinuity. The rest of the arguments works simultaneously (cf. \cite[Section 4.2]{hag11}).	

\item In general, Theorems \ref{Bthm9c4E} and \ref{Bthm4c8} can be extended to any locally compact \emph{Polish space} $X$, which is equipped with a non-atomic measure $m$. But, the conditions on the potential are much more complex and technical (cf. \cite[Chapter 5, esp. Theorem 5.2.8]{hag11}).
Moreover, the results presented here can be reformulated in the setting of configuration spaces over $\M \times \Xd$. Then, the class of of potentials and measures for which one can construct Gibbs perturbations on $\Gamma(\M \times \Xd)$ so far (cf. Remark \ref{rem41}) is considerably enriched.
All details, including the precise general framework, can be found in \cite[Chapters 4 and 5]{hag11}. 

\item The above results (including Theorems \ref{Bthm9c4E} and \ref{Bthm4c8}) also hold for any Lévy process on the cone $\K$, instead the Gamma process $\Gmea$ we mainly deal with in this paper. The only thing needed is the assumption that  the first two moments of the corresponding \emph{Lévy intensity measure} $\lambda$ are finite. In other words, $\lambda$ is a Radon measure on $\M$ satisfying
         $$
         	 \lambda(\M) = \infty,
         	 \quad \int_\M s \lambda(ds) < \infty
         	 \quad \text{and } \quad \int_\M s^2 \lambda(ds).$$

According to \cite{tsveyo01}, a \emph{{Lévy measure}} $\PmeaK \lambda$ on the cone $(\K, B(\K))$ over the location space $(\Xd,\B(\Xd),m)$ with Lévy (intensity) measure $\lambda$ on $\M$ is characterized by its Laplace transform  
        \begin{align}
	  {\mathbb {E}}_{\PmeaK \lambda} \left[ \exp \left(- \langle \varphi,\cdot\rangle \right) \right]
            = & \exp \left( - \int_{\M \times \X} (1 - e^{-\varphi (x)s}) \lambda(ds)  m (dx) \right),
       \notag \\ &
       	\qquad \qquad \quad 
       		\varphi \in C_0^+(\Xd)
       \label{2g2c1}.
        \end{align}
The existence of such $\PmeaK \lambda$ follows from the observation (cf. e.g. \cite{vegegr75} or \cite[Proposition 3.2]{kssu98}) that the Laplace transform \refeq{2g2c1} defines a so-called compound Poisson measure (with $\int_\M (1 - e^{-\varphi(x) s} ) \lambda(ds)$ being the Kolmogorov characteristic)  on $(\mathcal{D}'(\Xd), B(\Xd))$. Then, a standard analysis shows that this measure is supported by $\K$ and hence coincides with $\PmeaK \lambda$.

\item More generally, the construction scheme for Gibbs measures on $\pMa \supset \K$ outlined above works for any reference measure $\nu \in \mathcal{P}(\pMa)$ fulfilling the following properties (cf. \refeq{2z0c} and \refeq{eq27}):
\bitem 
	\item For all $\Delta \in \B_c(\Xd)$, we have 
	  \begin{align}
	   C_\Delta:= \int_{\pMa} \left[\eta(\Delta) + \eta(\Delta)^2\right] \nu(d\eta) < \infty.
	  \label{52c}
	  \end{align}
	  	\item For any $N \in \NN$, disjoint $\Delta_1, \dots, \Delta_N \in \B_c(\Xd)$ and $\varphi_1, \dots, \varphi_N \in L^\infty(\RR)$, it holds   
   \begin{align}
   	\int_{\K} \prod_{i=1}^N \varphi_i(\eta(\Delta_i)) \Gmea (d \eta)
   	= \prod_{i=1}^N \int_{\K} \varphi_i(\eta(\Delta_i)) \Gmea (d \eta)
   \label{52b}
 \end{align}
   \eitem
 Indeed, the first property ensures that Lemma \ref{Bl2b} holds and the second one that the local specification is consistent. Then the rest of the arguments also work. 
All random measures on $\mathcal{D}'(\Xd)$ obeying the \emph{independency} (or \emph{locality}) property \refeq{52b} are completely described by a \emph{Lévy-Khintchine}-type formula for their characteristic functionals (see Chapter III in \cite{gevi64d} for $d=1$ and Theorem 5 in \cite{rao68} for $d \geq 2$). According to Theorem 7.1 in \cite{kal83}, each random measure on $\Ma$ with independent increments is \emph{infinitely divisible}, apart from at most countable many fixed atoms. For a general theory of such so-called (completely) random measures on Polish spaces and their connection with Poisson processes, see Chapter 8 in \cite{kin93}.
\item Analogously to the case of Gamma measures (cf. \cite{hakove10}),  equilibrium diffusions  corresponding to theses Gibbs measures can be constructed via the general Dirichlet form approach for these Gibbs measures. For positive potentials this is presented in \cite[Theorem 7.4.4]{hag11}. For the case of general stable interactions see the forthcoming paper \cite{hakoly12}. This construction presumes additional information about the integrability of $\mu \in \tGibK \pot \theta \Xd$ (like those established in Theorem \ref{Bthm4c8}). 	
\item
It is an challenging open problem to extend Ruelle's technique of superstable estimates to Gibbs states on the cone $\K$. So far, this has been worked out on marked configuration spaces with bounded marks $s_i$'s or finite measures $\lambda(ds)$ on them (cf. \cite{mas00}). As it is seen from Lemma \ref{Bl2a}, to modify Ruelle's estimates in our setting it would be natural to use the local masses $\eta(Q_k) \in \RR_+$ instead of the counting mappings $|\tau(\eta_k)| \in \ZZ$.
\item
	In this paper we do not touch the uniqueness problem for $\mu \in \tGibK \pot \theta \Xd$. The required uniqueness can be archieved by small values of $\lbc$ or the shape parameter $\theta$. To this end, one can apply an appropriate modification of Dobrushin's uniqueness criterion, which will be done in a subsequent paper. An important ingredients of the uniqueness proof is the \guillemotleft one-point\guillemotright\ moment bound established in \refeq{45b}.
\eitem


\section{More properties of the Gamma measure}\label{s6}
For the interested reader, we give some more useful details concerning properties of the Gamma measures. Namely, we show that a Gamma measure is closely related to a certain Poisson measure and that it admits a Mecke type characterization.
\subsection{Explicit construction of the Gamma measures}
Let $m$ be a non-atomic Radon measure on $(\Xd,\B(\Xd))$. On $\M := (0,\infty)$ being equipped with the metric 
\begin{align}
d_\M (s_1,s_2):= \left| \log \frac {s_1}{s_2} \right|, \text{ $s_1, s_2 \in \M$}
\notag
\end{align}
and with the Borel $\sigma$-algebra $\B(\M)$ (that coincides, of course, with $\B(\RR) \cap \M$), we consider the Lévy measure
$$
  \lambda_\theta (dt) := \frac \theta  t e^{-t} dt
$$
with a fixed shaped parameter $\theta>0$. 
Each $\hat x = (s_x, x)$ from the product space $\hXd := \M \times \Xd$ may describe a particle with a mark $s_x$ being located at a position $x \in \Xd$.

In our considerations the configurations space $\Gamma(\hXd)$ over $\hXd$ will play a central role. It is defined as
\begin{align}
	\Gamma(\hXd) := \big\{ \gamma \subset \hXd \ \big| \  |\gamma_\Lambda| < \infty, \quad \forall \Lambda \in \B_c(\hXd) \big\}
	\notag,
\end{align}
where $\Lambda \in \B_c(\hXd)$ is a Borel set with compact closure and $|\gamma_\Lambda|$ denotes the set cardinality of $\gamma_\Lambda := \gamma \cap \Lambda$. 

It is well-known that $\Gamma(\hXd)$ is a Polish space in the vague topology inherited from $\mathbbm{M}(\hXd)$. By $\B(\Gamma(\hXd))$ we denote the associated Borel $\sigma$-algebra, which also coincides with the smallest $\sigma$-algebra generated by the cylinder sets $\{ \gamma \in \Gamma(\hXd) | \thinspace |\gamma_\Lambda | = n, \ \Lambda \in \B_c(\hXd) \}$, $n \in \mathbbm{Z}_+ := \NN \cup \{0\}$. 
 On $(\Gamma(\hXd),\B(\Gamma(\hXd)))$, we consider the Poisson measure $\Pmea \theta$ with the intensity measure $\lambda_\theta \otimes m$ on $(\hXd, \B(\hXd))$. It is given through its Laplace transform (see, e.g., \cite{alkoro98anaGeo,bechre76,gevi64d,kal83}),
\begin{align}&
\int\nolimits_{\mathit{\Gamma(\hXd) }}\exp \left\{ - \langle f,\gamma \rangle \right\}
 \Pmea \theta (d\gamma )
 \notag \\ & \quad =
 \exp \left\{ \int_{\hXd}\left(
e^{-f(\hat x)}-1\right) {\lambda_\theta \otimes m} (d \hat x)\right\} ,\text{ \ \ } f \in C_{0}^+(\hXd)
   \notag.
\end{align}
Here, we define for any $\B(\hXd)$-measurable function $f:\hXd \rightarrow  \RR_+ $ the pairing $\langle f, \gamma \rangle := \sum_{\hat x \in \gamma} f(\hat x)$.

We explain a constructive approach to define $\Pmea \theta$. 
For each $\Lambda \in \B_c(\hXd)$, the Poisson measure $\Pmead \theta \Lambda$ on 
\begin{align}
	\Csa \Lambda := \{ \gamma \in \Csa \hXd | \gamma \subset \Lambda\} = \bigsqcup_{n \geq 0} \{ \gamma \in \Csa \Lambda \thinspace | \thinspace | \gamma|=n \}
\label{7a}
\end{align}
with intensity measure $\lambda_\theta \otimes m$ is given by
$$
  \Pmead \theta \Lambda :=  e^{-\lambda_\theta \otimes m (\Lambda)} \sum_{n \geq 0} \frac 1 {n!} \big(\lambda_\theta \otimes m \big)^{\hat \otimes n}.
 $$
Due to the consistency of the family $\{\Pmead \theta \Lambda | \Lambda \in \B_c (\hXd) \}$, by Kolmogorov's theorem there exists a unique probability measure $\Pmea \theta$ on $(\Gamma(\hXd), \B(\Gamma(\hXd)))$ such that
$$
   \Pmea \theta \circ \mathbbm{P}_{\Lambda}^{-1} = \Pmead \theta \Lambda,
$$
where $\mathbbm{P}_{\Lambda}$ is the projection from $\Csa \hXd$ to $\Csa \Lambda$: $\mathbbm{P}_{\Lambda}(\gamma) := \gamma \cap \Lambda.$\\

	\begin{rem}\label{rem14f2}
	   	The \emph{Mecke} \emph{identity} (cf. Satz 3.1 in \cite{mec67}) is a useful characterization of the Poisson measure within the set of all probability measures $\nu$ in $\pMah$ having finite first local moments, i.e., 
	   	$$
	   	\int \langle  \mathbbm{1}_\Lambda, \gamma \rangle \nu(d\gamma) < \infty
	   	\text{  
	   	for all $\Lambda \in \B_c(\hXd)$}.
	   $$
		Let $F : \hXd \times {\pMah} \rightarrow \RR_+$ be a $\mathcal{B}(\hXd) \times \mathcal{B}(\pMah)$-measurable function. Then
		\begin{align}
			&\int_{\pMah} \int_{\hXd} F(\hat x,\gamma) \gamma(d\hat x) \Pmea \theta (d\gamma)
		\notag \\ & \qquad =
			\int_{\hXd} \int_{\pMah} F(\hat x, \gamma + \delta_{\hat x}) \Pmea \theta (d\gamma) {\lambda_\theta \otimes m}(d\hat x)
		\notag
		\end{align}
		and the family of these equations with all such $F$'s uniquely determines $\Pmea \theta$.
		\end{rem}

We identify a smaller set $\Gamma_f (\hXd) \subset \Gamma (\hXd)$ that carries $\Pmea \theta$: To that end, we introduce the set of \emph{pinpointing} configurations
\begin{align}
					\Gamma_p(\hXd) := \big\{ \gamma \in \Gamma(\hXd) \big|& \text{ $\forall$ $(s_1,x_1), (s_2,x_2) \in \gamma:$}
					\notag \\& \quad
							x_1 = x_2 \Rightarrow s_1 = s_2 \big\}
					\notag. 
				\end{align}
For all $\gamma \in \Gamma_p(\hXd)$ and $\Delta \in \B_c(\Xd)$ we define a local mass via 
			\begin{align}
				\lomagf {\Delta} \gamma := \sum_{\hat x=(s_x,x) \in \gamma} s_x \mathbbm{1}_\Delta (x)
				  = \int_{\hXd} s \mathbbm{1}_\Delta (x) \gamma( d \hat x)
				\notag. 
			\end{align}
Combining these two definitions, we specify the set of pinpointing configurations with \emph{finite local mass} 
			\begin{align}
				\Csfsad \lomagsy \hXd := \big\{ \gamma \in \Gamma_p(\hXd) \thinspace \big| \thinspace \lomagf {\Delta} \gamma < \infty, \ \forall \Delta \in \B_c(\Xd) \big\}
				\label{in7z0e}
			\end{align}
and the corresponding trace $\sigma$-algebra $\B(\Gamma_f(\hXd):= \{ A \cap \Gamma_f(\hXd) | A \in \B(\Gamma(\hXd))\}$.
An important observation (cf. \cite[Proposition 3.4]{kssu98} and \cite[Theorems 2.2.4 and 2.2.9]{hag11}), which enables us to construct the Gamma measure, is 
\begin{align}
\Pmea \theta (\Csfsad \lomagsy \hXd) = 1.
\label{714b}
\end{align}
For our further considerations it is crucial that there exists a \emph{bijection} between $\Csfsad \lomagsy \hXd$ and the cone of discrete Radon measures $\Ka \Xd$ (cf. Eq. \refeq{2z0a}).  
This bijection, being even a homeomorphism, is defined by
\funcL  {{\mathbbm{T}}}   {\Gamma_f ({\hXd})} {\Ka \Xd}
                        {\gamma = \{(s_x,x)\}}       {\eta := \sum_{(s_x,x) \in \gamma} s_x \delta_{x}.}
{714c}

\begin{thm}\label{thm5a1}
	The image $\sigma$-algebra of $\B(\Gamma_f(\hXd))$ and $\B(\K)$ coincide, i.e.,
	$$
		\B(\K) = \left\{ {{\mathbbm{T}}} ( A \cap \Gamma_f(\hXd)) \big| A \in \B(\Csa \hXd) \right\}.
	$$
	Moreover, the image measure $\Pmea \theta \circ \mathbbm{T}^{-1}$ of (the restriction of) $\Pmea \theta$ (to $\Gamma_f(\hXd)$) under ${\mathbbm{T}}$ coincides with  the \emph{Gamma} measure $\PmeaK \theta$, i.e., for all bounded measurable $F: \K \rightarrow \RR$ 
	\begin{align}
		\int_{\K} F(\eta) \Gmea (d \eta) = \int_{\Gamma_f(\hXd)} F\big(\mathbbm{T} (\gamma) \big)  \Pmea \theta (d \gamma).
	\label{7a20}
	\end{align}
	In particular, for $g \in C_b^\infty(\RR^{N})$, $\varphi_1, \dots, \varphi_N \in C_0^\infty(\Xd)$ and $N \in \NN$
	\begin{align}
		\int_{\Ka \Xd}& g\left( \langle \varphi,\eta \rangle, \dots, \langle \varphi_N, \eta \rangle \right) \Gmea(d\eta) 
		\notag \\ &=
			\int_{\Gamma_f ({\hXd})} g \left( \langle s\otimes \varphi_1, \gamma \rangle, \dots, \langle s \otimes \varphi_N, \gamma \rangle \right) \Pmea \theta (d\gamma)
			\notag.
			\end{align}
\end{thm}
\begin{proof}
For the details, cf. \cite{hag11,hakove10}.
The main idea to prove the first property is to apply subsequently Kuratowski's theorem (cf. \cite[Theorem V.2.4]{par67})  and to use that $\Gamma_f (\hXd) \in \B_c(\Csa \hXm))$ (cf. \cite[Theorems 2.2.4 and 2.2.9]{hag11}) equipped with the trace $\sigma$-algebra is a separable Borel space (cf. , \cite[Theorem V.2.2]{par67}).

For the second property, by \cite[Theorem 3.7]{bechre76}, it is sufficient to check that the Laplace transforms coincide.
\end{proof}

Note that by \refeq{7a} and \refeq{714c}
		\begin{align}
			\Gmea \left( \{ \eta \in \K | \tau(\eta_\Delta) = n \} \right) = 0,
			\qquad \forall n \in \mathbbm{Z}_+, \ \forall \Delta \in \B_c(\Xd).
		\notag
		\end{align}
		But (cf. \refeq{eq24b}), 
		$$
			\E_\Gmea \left[ \eta(\Delta)^m \right] < \infty, 
			\qquad \forall m \in \NN, \  \Delta \in \B_c(\Xd).
		$$
This immediately follows (cf. Eqs. \refeq{in7z0e} and \refeq{7a20}) from the following properties of the intensity measure $\lambda \otimes m$: for any $0<s_1<s_2< \infty$, $x \in \Xd$ and $\Delta \in \B_c(\Xd)$ 
$$
	\lambda_\theta \otimes m \big( (s_1,s_2) \times \{x\} \big) = 0 
	\quad\text{and} \quad
	\int_{\hXd} s \mathbbm{1}_\Delta \lambda_\theta \otimes m (ds,dx) < \infty.
$$
		
\subsection{Mecke and GNZ identities}
Similar as for the Poisson measure (cf. Remark \ref{rem14f2}), we have a Mecke-type characterization for the Gamma measures.

\begin{thm}\label{thm64}
Fix $\theta > 0$ and let $\nu \in \mathcal{P}(\pMa)$ have \emph{finite first local moments}, i.e., for each $\Delta \in \B_c(\Xd)$
$$
	\int_{\pMa} \eta(\Delta) \nu(d\eta) < \infty.
$$ 
Then the following are equivalent:
\bitem
 \item The measure $\nu$ is a Gamma measure, i.e., $\nu = \Gmea$. 
 \item 
 		For any measurable function $F : \Xd \times \pMa \rightarrow \RR_+$
		\begin{align}
			&\int_\pMa \int_{\Xd} F(x,\eta) \eta(dx) \nu (d\eta)
		\notag \\ & \qquad =
			 \int_{{\pMa}}  \int_{\Xd} \int_{\M} s F(x, \eta + s \delta_{x})  {\lambda_\theta (ds)} m(dx) \nu (d\eta)
		\label{eq5e}.
		\end{align}
\eitem
\end{thm}

\begin{proof}
By Remark \ref{rem14f2} and Theorem \ref{thm5a1}, the first property implies the second one. Indeed, consider the functions $F$ of the special form 
\begin{align}
	F(x, \eta) := f(x) g(\langle \varphi_1, \eta\rangle, \dots, \langle \varphi_N, \eta \rangle)
	\label{14c2}
\end{align}
with $f, \ \varphi_1, \dots, \varphi_N \in C_0(\Xd)$, $g \in C_0(\RR^N)$ and $N \in \NN$. Then one can rewrite the left-hand side of Eq. \refeq{eq5e} for $\nu = \Gmea$ as 
\begin{align}&
	\int_\K \int_{\Xd} F(x,\eta) \eta(dx) \Gmea (d\eta)
	\notag \\ =&
		\int_\K \langle f, \eta \rangle g\left(\langle \varphi_1, \eta\rangle, \dots, \langle \varphi_N, \eta \rangle \right) \Gmea (d\eta)
	\notag \\ =&
		\int_{\Gamma(\hXd)} \langle s \otimes f, \gamma \rangle g \left( \langle s \otimes \varphi_1, \gamma \rangle, \dots \langle s \otimes \varphi_N, \gamma \rangle \right) \Pmea \theta (d \gamma)
	\notag \\ =&
		\int_{\Xd} \int_{\M} \int_{\Gamma (\hXd)} 
			g \left( \langle s \otimes \varphi_1, \gamma + \delta_{(s,x)} \rangle, \dots, \langle s \otimes \varphi_N, \gamma + \delta _{(s,x)} \rangle \right) 
		\notag \\ & \qquad \qquad \qquad 
			\times s f(x) \Pmea \theta (d\gamma)  \lambda_\theta (ds) m(dx)
	\notag \\ =&
		\int_{\Xd} \int_{\M} \int_{\K} 
			g \left( \langle \varphi_1, \eta + s \delta_x \rangle, \dots, \langle \varphi_N, \eta + s \delta_x \rangle \right) f(x) \Gmea (d\eta) s \lambda_\theta (ds) m (dx)
	\notag \\ =&
		\int_{\K} \int_{\Xd} \int_{\M} F(x, \eta +s \delta_x) s \lambda_\theta (ds) m(dx) \Gmea(d\eta)
	\notag.
	\end{align}
	
	Since the cylinder functions of the form \refeq{14c2} generate the product $\sigma$-algebra $\B(\Xd) \otimes \B(\K)$, by the monotone class theorem (see e.g., Theorem I.8 in \cite{pro05}) the identity \refeq{eq5e} extends to all bounded $\B(\Delta) \otimes \B(\K)$-measurable $F$ with arbitrary $\Delta \in \B_c(\Xd)$. Note that for such $F$ all integrals above are finite by \refeq{eq24b}. Finally, using a standard cut-off argument and the Beppo Levi monotone convergence theorem, one proves the validity of \refeq{eq5e} for all $\B(\Xd) \otimes \B(\K))$-measurable $F \geq 0$.\\

For the other direction, we employ Eq. \refeq{eq5e} to characterize the involved measure $\nu \in \mathcal{P}(\pMa(\Xd))$ by identifying its Laplace transform. 
Indeed, fix $\varphi \in \fcone$ and define
\begin{align}
	\RR_+ \ni t \mapsto L(t):= \int_{\pMa} \exp \left\{ - t \langle \varphi, \eta \rangle \right\} \nu(d\eta) \in  \RR_+ 
	\notag.
\end{align}
Using that all local moments are finite, Young's inequality and Lebesgue's dominated convergence theorem, we see that the function $L$ is strictly positive, continuous on $[0,\infty)$ and continuously differentiable on $ (0,\infty)$. 
Then we have, differentiating and using Eq. \refeq{eq5e} for $t > 0$,
\begin{align}
	\frac d {dt} L(t) &= - \int_{\Xd} \int_\M \int_{\K} 
			s \exp \left\{ - t\langle \varphi, \eta + s \delta_x \rangle \right\}
			\varphi(x)  \mu(d\eta) \lambda_\theta (ds) m(dx)
			\notag \\ & 
				= - C(t) L(t)
			\notag,
\end{align}
where the continuous function $\RR_+ \ni t \mapsto C(t) \in \RR_+$ is defined by
\begin{align}
		C(t) :=\int_\M \int_{\Xd} s \varphi(x) e^{-t s \varphi(x)} \lambda_\theta (ds) m(dx)
			\notag.
\end{align}
This is an exactly solvable equation of Gronwall type. Its unique continuous solution with the initial data $L(0)=1$ is given by 
\begin{align}
	L(t) &= L(0) \exp \left\{ -\int_0^t C(r) dr \right\}
	\notag \\ &=
		 \exp \left\{ -\int_{\M} \int_{\Xd} \left(e^{-s \varphi(x)} - 1 \right) \lambda_\theta (ds) m (dx) \right\}
		 \notag.
\end{align}
Since the latter holds for all $\varphi \in \fcone$, we get by the uniqueness of the Laplace transform that $\nu = \Gmea$. 
\end{proof}

For all Gibbs measures $\mu \in \tGibK \pot \theta \Xd$, an equation similar to \refeq{eq5e} is also true. It is called Georgii-Nquen-Zessin identity (GNZ for short) and was first established on configuration spaces in \cite{geo76} and \cite{ngze79}.

\begin{thm}
For any measurable function $F: \Xd \times \K \rightarrow \RR_+$ and any $\mu \in \tGibK \pot \theta \Xd$
\begin{align}
	&\int_{\K} \int_{\Xd} F(x,\eta) \eta(dx) \mu(d\eta) 
	= 
		\int_{\K} \sum_{\substack{x \in \tau(\eta)\\ \eta = ( (s_x,x))_{x \in \tau(\eta)}}} s_x F(x,\eta) \mu(dx) 
	\notag \\ =&
		\int_{\K} \int_{\hXd} F(x, \eta + s \delta_x ) e^{-\Phi ((s,x);\eta)}
		s \lambda_\theta (ds) m (dx) \mu(d\eta)
	\label{eq65a},
	\end{align}
	where for $\eta := (s_y,y)_{y \in \tau(\eta)} \in \K$
	$$
		\Phi\big( (s,x); \eta \big) := 2 s \sum_{y \in \tau(\eta)} s_y \phi(x,y).
	$$	
\end{thm}

\begin{proof}
As was explained in the proof of Theorem \ref{thm64}, it suffices to establish \refeq{eq65a} for all functions $F$ of the form $F(x,\eta) := f(x) g(\eta_\Delta)$, where the support of $f \in C_0(\Xd)$ lies in $\Delta \in \B_c(\Xd)$ and $g: \Ka \Delta \rightarrow \RR$ is bounded and measurable. 

Then by the DLR equation \refeq{Beq10ab} and the Mecke identity \refeq{eq5e}
\begin{align}&
	\int_{\K} \int_{\Xd} F(x,\eta) \eta(dx) \mu(d\eta) 
	= 
		\int_\K \langle f, \eta_\Delta \rangle g\left(\eta_\Delta \right) \mu(d\eta) 
		\notag \\ =&
			\int_\K \int_\K \langle f, \eta_\Delta \rangle g \left( \eta_\Delta \right) \pi_\Delta (d \eta| \xi) \mu( d\xi)
	\notag \\ =&
	\int_{\K} \int_{\Ka \Delta} \langle f, \eta_\Delta \rangle g(\eta_\Delta) 
  	\frac 1 {Z_\Delta(\xi)} e^{-H(\eta_\Delta | \xi_{\Delta^c})} \PmeaKd \theta \Delta  (d \eta_\Delta) \mu(d\xi)
	\notag \\ =&
		\int_{\K} \int_{\Ka \Delta} \int_\Delta \int_\M 
			f(x) g\left( \eta_\Delta + s \delta_x \right) 			
			\frac 1 {Z_\Delta(\xi)} 
	\notag \\ & \qquad \qquad 
			\times \exp \left\{-H(\eta_\Delta + s \delta_x | \xi_{\Delta^c})\right\}
			s \lambda_\theta (ds) m(dx) \PmeaKd \theta \Delta (d \eta_\Delta) \mu(d\xi)
	\notag \\ =&
		\int_\Delta \int_\M \int_{\K} \int_{\Ka \Delta} 
			F(x,\eta_\Delta + s \delta_x) 
			\frac 1 {Z_\Delta(\xi)} 
		\exp\left\{- H(\eta_\Delta| \xi_{\Delta^c} )\right\} 
	\notag \\ & \qquad \qquad 
				\times \exp \left\{-\Phi\big( (s,x); \eta_\Delta + \xi_{\Delta^c}\big)\right\} 
				\PmeaKd \theta \Delta (d \eta_\Delta) \mu(d \xi) s \lambda_\theta (ds) m(dx) 
	\notag \\
	=&
		\int_\Delta \int_\M \int_{\K} \int_{\K} 
			F(x, \eta+ s \delta_x) e^{-\Phi\big( (s,x); \eta \big)} \pi_\Delta(\eta|\xi) \mu(d\xi) s \lambda_\theta (ds) m(dx)
	\notag.  \intertext{The last line equals} &
		\int_\Delta \int_\M \int_\K F(x, \eta + s \delta_x) e^{-\Phi\big( (s,x);\eta \big)} \mu(d\xi) s \lambda_\theta(ds) m (dx)
	\notag \\ =&
		\int_\K \int_\Xd \int_\M  F(x, \eta + s \delta_x) e^{-\Phi\big( (s,x);\eta \big)} s \lambda_\theta(ds) m (dx) \mu(d\eta),
		\notag
\end{align}
which proves the GNZ identity \refeq{eq65a}.
\end{proof} 	 

\subsection{FKG inequality}
We introduce a partial order on the cone $\K$. For any two measures 
$$
	\eta= \big( (s_x,x) \big)_{x \in \tau(\eta)}, \thinspace
	\eta' = \big( (s_x',x') \big)_{x' \in \tau(\eta)} \in \K
$$
we write $\eta \leq \eta'$ if $\eta(\Delta) \leq \eta'(\Delta)$ for all $\Delta \in \B_c(\Xd)$. In the language of particles, this means that $\tau(\eta) \subseteq \tau(\eta')$ and $s_x \leq s_x'$ for each $x \in \tau(\eta)$. A function $F: \K \rightarrow \RR$ is called \emph{increasing} if 
\begin{align}
	F(\eta) \leq F(\eta') \quad \text{whenever} \quad \eta \leq \eta'.
	\notag
\end{align}
A typical example of such $F$ is given by the following cylinder functions 
\begin{align}
	F(\eta) := f\left( \langle \varphi_1, \eta \rangle, \dots, \langle \varphi_N, \eta \rangle \right), 
	\label{63a}
\end{align}
where $\varphi_1, \dots, \varphi_N \in C_0^+(\Xd)$, $N \in \NN$ and $f: \RR^N \rightarrow \RR$ is monotonically increasing in each argument. 

\begin{prop}[FKG inequality]\label{prop66}
The Gamma measure $\Gmea$ obeys the FKG inequality, which says that $\Gmea$ has \emph{positive correlations}
\begin{align}
	\text{Cov}_{\Gmea} (F,G) := \int_\K F G d \Gmea - \int_\K F d \Gmea \int_\K G d \Gmea \geq 0
	\notag
\end{align}
for all bounded increasing measurable functions $F,G: \K \rightarrow \RR$.
\end{prop}

By the monotone or dominated convergence, the result immediately extends to unbounded functions provided $F,G \geq 0$ or $F,G \in L^2(\K,\Gmea)$.\\

The FKG inequality is well-known for Poisson measures on configuration spaces (see Lemma 2.1 in \cite{jan84} and Corollary 1.2 in \cite{geku97}) or more generally, for infinitely divisible $\Ma$-valued random variables (see Theorem 1.1 in \cite{eva90} and \cite{buwa85}). In statistics, the random measures satisfying the FKG inequality are called \emph{associated}.

\begin{proof}[Proof of Proposition \ref{prop66}]
We use the identity $\Gmea = \Pmea \theta \mathbbm{T}^{-1}$ and the FKG inequality for $\Pmea \theta$ on $\Gamma(\hXd)$. Consider a pair of bounded monotone functions $F,G: \K \rightarrow \RR$. Then $\hat F := F \circ \mathbbm{T}$ and $\hat G := G \circ \mathbbm{T}$ are monotone functions on $\Gamma(\hXd)$. To this end, note that the homemorphism $\mathbbm{T}: \Gamma_f(\Xd) \rightarrow \K$ is order preserving, i.e., $\gamma_1 \geq \gamma_2$ in $\Gamma_f (\hXd)$ implies $\mathbbm{T} \gamma_1 \geq \mathbbm{T} \gamma_2$ in $\K$. The latter is equivalent to checking that for any $\varphi \in C_0^+(\Xd)$ 
$$
	\langle \varphi, \mathbbm{T} \gamma_1 \rangle = \langle s \otimes \varphi, \gamma_1 \rangle \geq \langle s \otimes \varphi, \gamma_2 \rangle = \langle \varphi, \mathbbm{T} \gamma_2 \rangle.
$$
Therefore, we have 
\begin{align}&
\int_\K F(\eta) G(\eta) \Gmea (d \eta) 
= 
	\int_{\Gamma(\hXd)} F (\mathbbm{T} \gamma) G(\mathbbm{T} \gamma) \Pmea \theta (d\gamma) 
	\notag \\ =&
		\int_{\Gamma_f(\hXd)} \hat F(\gamma) \hat G (\gamma) \Pmea \theta (d\gamma)
	\geq 
		\int_{\Gamma_f (\hXd)} \hat F (\gamma) \Pmea \theta (d \gamma)
		\int_{\Gamma_f (\hXd)} \hat G (\gamma) \Pmea \theta (d \gamma)
	\notag \\ =&
		\int_{\K} F(\eta) \Gmea (d \eta) \int_\K G(\eta) \Gmea (d \eta),
	\notag
\end{align}
which yields the result.
\end{proof}

\begin{rem}
	There is a standard way of extending FKG correlation inequalities to ferromagnetic models. However, pure attractive pair interactions are not physically relevant for particle systems in the continuum. In \cite{got05} FKG inequalities have been proven for particle systems on marked configuration spaces with so-called weakely attractive interactions and then used to study existence and uniqueness of the corresponding Gibbs states.
\end{rem}

\section*{Acknowledgements}
We thank Eugene Lytvynov, Ilya Molchanov and Anatoly Vershik for valuable discussions. Financial support by the DFG through the  CRC (SFB) 701 ``Spectral Structures and Topological Methods in Mathematics'' and the IRTG (IGK) 1132 ``Stochastics and Real World Models'' is gratefully acknowledged.


\begin{thebibliography}{10}

\bibitem{alkoro98anaGeo}
S.~Albeverio, Yu.~G. Kondratiev, and M.~R\"ockner.
\newblock Analysis and geometry on configuration spaces.
\newblock {\em J. Funct. Anal.}, 154(2):444--500, 1998.

\bibitem{alkoro98anaGibs}
S.~Albeverio, Yu.~G. Kondratiev, and M.~R\"ockner.
\newblock Analysis and geometry on configuration spaces: The {G}ibbsian case.
\newblock {\em J. Funct. Anal.}, 157:242--291, 1998.

\bibitem{alba81}
D.~J. Aldous and M.~T. Barlow.
\newblock On countable dense random sets.
\newblock In {\em Seminar on {P}robability, {XV} ({U}niv. {S}trasbourg,
  {S}trasbourg, 1979/1980) ({F}rench)}, volume 850 of {\em Lecture Notes in
  Math.}, pages 311--327. Springer, Berlin, 1981.

\bibitem{bara97b}
T.~O. Banakh and T.~N. Radul.
\newblock Topology of spaces of probability measures.
\newblock {\em Sb. Math.}, 188(7):973--995, 1997.

\bibitem{bechre76}
Christian Berg, Jens Peter~Reus Christensen, and Paul Ressel.
\newblock Positive definite functions on abelian semigroups.
\newblock {\em Math. Ann.}, 223(3):253--274, 1976.

\bibitem{buwa85}
R.~Burton and E.~Waymire.
\newblock Scaling limits for associated random measures.
\newblock {\em Ann. Probab.}, 13(4):1267--1278, 1985.

\bibitem{dob70b}
R.~L. Dobrushin.
\newblock Gibbsian random fields for particles without hard core.
\newblock {\em Teoret. Mat. Fiz.}, 4(1):101--118, 1970.

\bibitem{dob70}
R.~L. Dobrushin.
\newblock Prescribing a system of random variables by conditional distribtions.
\newblock {\em Theory Probab. Appl.}, 15:459--405, 1970.

\bibitem{eva90}
Steven~N. Evans.
\newblock Association and random measures.
\newblock {\em Probab. Theory Related Fields}, 86:1--19, 1990.
\newblock 10.1007/BF01207510.

\bibitem{gevi64d}
I.M. Gelfand and N.Ya. Vilenkin.
\newblock {\em Generalized {F}unctions}, volume~4.
\newblock Academic Press, 1964.

\bibitem{geku97}
H.-O. Georgii and T.~K{\"u}neth.
\newblock Stochastic comparison of point random fields.
\newblock {\em J. Appl. Probab.}, 34(4):pp. 868--881, 1997.

\bibitem{geo76}
Hans-Otto Georgii.
\newblock Canonical and grand canonical {G}ibbs states for continuum systems.
\newblock {\em Comm. Math. Phys.}, 48(1):31--51, 1976.

\bibitem{geo79}
Hans-Otto Georgii.
\newblock {\em Canonical {G}ibbs measures}, volume 760 of {\em Lecture {N}otes
  in {M}athematics}.
\newblock Springer, Berlin, 1979.
\newblock {S}ome extensions of de {F}inetti's representation theorem for
  interacting particle systems.

\bibitem{geo88}
Hans-Otto Georgii.
\newblock {\em Gibbs {M}easures and {P}hase {T}ransitions}, volume~9 of {\em de
  {G}ruyter {S}tudies in {M}athematics}.
\newblock Walter de Gruyter \& Co., Berlin, 1988.

\bibitem{got05}
Hanno Gottschalk.
\newblock Particle systems with weakly attractive interaction.
\newblock {\em Methods Funct. Anal. Topology}, 11(4):356--369, 2005.

\bibitem{hag11}
D.~Hagedorn.
\newblock {\em {S}tochastic {A}nalysis related to {G}amma measures - {G}ibbs
  perturbations and associated {D}iffusions}.
\newblock PhD thesis, Universit\"{a}t Bielefeld, 2011.

\bibitem{hakoly12}
D.~Hagedorn, Yu.~G. Kondratiev, and Eugene Lytvynov.
\newblock Quasi-invariance of {D}irichlet forms related to {G}ibbs
  perturbations of {G}amma measures.
\newblock {U}niversit\"{a}t {B}ielefeld, 2012.

\bibitem{hakove10}
D.~Hagedorn, Yu.~G. Kondratiev, Eugene Lytvynov, and A.~M. Vershik.
\newblock Integration by parts formula for the gamma process.
\newblock preprint, {U}niversit\"{a}t {B}ielefeld, 2012.

\bibitem{jan84}
Svante Janson.
\newblock Bounds on the distributions of extremal values of a scanning process.
\newblock {\em Stoch. Processes Appl.}, 18(2):313 -- 328, 1984.

\bibitem{kal83}
O.~Kallenberg.
\newblock {\em Random {M}easures}.
\newblock Akad.-Verl., Berlin, 1983.

\bibitem{ken00}
Wilfrid~S. Kendall.
\newblock Stationary countable dense random sets.
\newblock {\em Adv. Appl. Probab.}, 32(1):86--100, 2000.

\bibitem{kin93}
J.~F.~C. Kingman.
\newblock {\em {P}oisson Processes}.
\newblock Clarendon Press, Oxford, 1993.

\bibitem{kosist98}
Yu.~G. Kondratiev, J.~L. de~Silva, and L.~Streit.
\newblock Differential geometry on compound {P}oisson space.
\newblock {\em Methods Funct. Analysis Topology}, 4(1):32--58, 1998.

\bibitem{kssu98}
Yu.~G. Kondratiev, J.~L. de~Silva, L.~Streit, and G.~F. Us.
\newblock Analysis on {P}oisson and {G}amma spaces.
\newblock {\em Infinite Dimensi. Anal., Quantum Probab. Relat. Topi.},
  1(1):91--117, 1998.

\bibitem{koparo10}
Yu.~G. Kondratiev, T.~Pasurek, and M.~R{\"o}ckner.
\newblock Gibbs measures of continuous systems: An analytic approach.
\newblock submitted to Reviews. Math. Physics (2012), 2010.

\bibitem{kunphd}
T.~Kuna.
\newblock {\em Studies in configuration space analysis and applications}.
\newblock PhD thesis, Rheinische Friedrich-Wilhelms-Universit{\"{a}}t Bonn,
  1999.

\bibitem{kokusi98}
T.~Kuna, Yu.~G. Kondratiev, and J.~L. de~Silva.
\newblock Marked {G}ibbs measures via cluster expansion.
\newblock {\em Methods Funct. Anal. Topology}, 4(4):50--81, 1998.

\bibitem{laru69}
O.~E. Lanford, III and D.~Ruelle.
\newblock Observables at infinity and states with short range correlations in
  statistical mechanics.
\newblock {\em Comm. Math. Phys.}, 13:194--215, 1969.

\bibitem{mas00}
Shigeru Mase.
\newblock Marked {G}ibbs processes and asymptotic normality of maximum
  pseudo-likelihood estimators.
\newblock {\em Math. Nachr.}, 209:151--169, 2000.

\bibitem{mec67}
J.~Mecke.
\newblock Stationäre zufällige {M}a{\ss}e auf lokalkompakten {A}belschen
  {G}ruppen.
\newblock {\em Z. {W}ahrsch. verw. {G}ebiete}, 9:36--58, 1967.

\bibitem{mei34}
J.~Meixner.
\newblock Orthogonale {P}olynomsysteme mit einer besonderen {G}estalt der
  erzeugenden {F}unktion.
\newblock {\em J. London Math. Society}, 9:6--13, 1934.

\bibitem{ngze79}
Xuan-Xanh Nguyen and Hans Zessin.
\newblock Integral and differential characterizations of the {G}ibbs process.
\newblock {\em Math. Nachr.}, 88:105--115, 1979.

\bibitem{par67}
K.R. Pathasarathy.
\newblock {\em Probalistic measures on metric spaces}.
\newblock New York-London: Academic Press, 1967.

\bibitem{pre76}
Ch.~J. Preston.
\newblock {\em Random {F}ields}.
\newblock Lecture {N}otes in {M}athematics; 534. Springer, 1976.

\bibitem{pre05}
Ch.~J. Preston.
\newblock Specifications and their {G}ibbs states.
\newblock Lecture notes, Universit\"{a}t Bielefeld, available online at www.
  math.uni-bielefeld.de/\symbol{126}preston/rest/gibbs/files/specifications.pdf,
  2005.

\bibitem{pro05}
Ph.~E. Protter.
\newblock {\em Stochastic {I}ntegration and {D}ifferential {E}quations}.
\newblock Stochastic {M}odelling applied {P}robab.; 21. Springer, 2005.

\bibitem{rao68}
M.~M. Rao.
\newblock Local functionals and generalized random fields.
\newblock {\em Bull. Amer. Math. Soc.}, 74:288--293, 1968.

\bibitem{rue69}
D.~Ruelle.
\newblock {\em Statistical {M}echanics: {R}igorous {R}esults}.
\newblock W. A. Benjamin, Inc., New York-Amsterdam, 1969.

\bibitem{rue70}
D.~Ruelle.
\newblock Superstable interactions in classical statistical mechanics.
\newblock {\em Comm. Math. Phys.}, 18:127--159, 1970.

\bibitem{sta03}
W.~Stannat.
\newblock Spectral properties for a class of continuous state branching
  processes with immigration.
\newblock {\em J. Funct. Anal.}, 201(1):185 -- 227, 2003.

\bibitem{tsveyo01}
N.~Tsilevich, A.~M. Vershik, and M.~Yor.
\newblock An infinite-dimensional analogue of the {L}ebesgue measure and
  distinguished properties of the {G}amma process.
\newblock {\em J. Funct. Anal.}, 185(1):274 -- 296, 2001.

\bibitem{ver07c}
A.~M. Vershik.
\newblock Does a {L}ebesgue measure in an infinite-dimensional space exist?
\newblock {\em Tr. Mat. Inst. Steklova}, 259(Anal. i Osob. Ch. 2):256--281,
  2007.

\bibitem{vegegr75}
A.~M. Ver{\v{s}}ik, I.~M. Gel{\cprime}fand, and M.~I. Graev.
\newblock Representations of the group of diffeomorphisms.
\newblock {\em Uspehi Mat. Nauk}, 30(6(186)):1--50, 1975.

\end{thebibliography}

\def\cprime{$'$} \def\cprime{$'$}

\end{document}